\documentclass[11pt]{article}
\usepackage[utf8]{inputenc}
\usepackage[margin=1in]{geometry}

\usepackage[rgb,table,usenames,dvipsnames]{xcolor}
\usepackage{pdfpages}

\usepackage{setspace}
\usepackage{verbatim}
\usepackage{subcaption}
\usepackage{boxedminipage}
\usepackage{csquotes}
\usepackage{fancyhdr}
\usepackage{graphicx}
\usepackage{hdb_macros}
\usepackage{fullpage}
\usepackage{algpseudocode}
\usepackage{braket}
\usepackage{stmaryrd}

\usepackage[draft,multiuser,inline,nomargin]{fixme}
\FXRegisterAuthor{n}{en}{\color{blue}Venkat}
\FXRegisterAuthor{s}{es}{\color{red}Surendra}
\fxusetheme{color}

\newcommand{\fullversion}[1]{#1}
\newcommand{\fullornot}[2]{#1}

\usepackage{xcolor}
\hypersetup{
	colorlinks,
	linkcolor={red!75!black},
	citecolor={blue!75!black},
	urlcolor={red!75!black},
}
 \usepackage[hyperpageref]{backref}

\usepackage{esvect}

 \renewcommand{\C}{\mathcal{C}}
\newcommand{\xorassump}{planted XOR assumption}

\title{New constructions of pseudorandom codes}
\author{Surendra Ghentiyala \thanks{Cornell University.  \email{sg974@cornell.edu}. This work is supported in part by the NSF under Grants Nos.~CCF-2122230 and CCF-2312296, a Packard Foundation Fellowship, and a generous gift from Google. This work was done while the author was visiting the Simons Institute for the Theory of Computing.}
\and Venkatesan Guruswami \thanks{Simons Institute for the Theory of Computing and Department of EECS, University of California, Berkeley. \email{venkatg@berkeley.edu}. Research supported in part by a Simons Investigator award, and NSF grant CCF-2211972.}}
\date{May 2025}

\usepackage{tikz}
\usepackage{pgfplots}
\pgfplotsset{width=10cm,compat=1.10}
\usepgfplotslibrary{fillbetween}
\usetikzlibrary{patterns}
\usetikzlibrary{patterns.meta}
\allowdisplaybreaks

\begin{document}
\pagenumbering{roman}

\maketitle

\listoffixmes

\begin{abstract}
    Introduced in \cite{christ2024pseudorandom}, pseudorandom error-correcting codes (PRCs) are a new cryptographic primitive with applications in watermarking generative AI models. These are codes where a collection of polynomially many codewords is computationally indistinguishable from random for an adversary that does not have the secret key, but anyone with the secret key is able to efficiently decode corrupted codewords. In this work, we examine the assumptions under which PRCs with robustness to a constant error rate exist.
\begin{enumerate}
    \item 
    We show that if both the planted hyperloop assumption introduced in \cite{BKR} and security of a version of Goldreich's PRG hold, then there exist public-key PRCs for which no efficient adversary can distinguish a polynomial number of codewords from random with better than $o(1)$ advantage.
    \item 
    We revisit the construction of \cite{christ2024pseudorandom} and show that it can be based on a wider range of assumptions than presented in \cite{christ2024pseudorandom}. To do this, we introduce a weakened version of the planted XOR assumption which we call the weak planted XOR assumption and which may be of independent interest.
    \item 
    We initiate the study of PRCs which are secure against space-bounded adversaries. We show how to construct secret-key PRCs of length $O(n)$ which are \emph{unconditionally} indistinguishable from random by $\poly(n)$ time, $O(n^{1.5-\varepsilon})$ space adversaries.
\end{enumerate}
\end{abstract}

\thispagestyle{empty}
\newpage
\tableofcontents
\newpage
\pagenumbering{arabic}

\section{Introduction}
\label{sec:intro}
The ability of malicious actors to easily and cheaply generate large amounts of AI generated content is becoming a larger issue as generative AI models progress. Digital watermarking mitigates some of these concerns by offering a way to generate AI content which is later easily recognizable (to someone with the secret key) as AI generated. One may also hope to recover other information possibly embedded in the watermark at the time of creation, like date of generation. Cryptographic watermarking leverages cryptography to give watermarking schemes with provable guarantees, as opposed to ad-hoc schemes. In this work, we expand the set of assumptions under which one can achieve cryptographic watermarking.

\subsection{PRCs and Applications}
An exciting recent work by Christ and Gunn \cite{christ2024pseudorandom} introduced the notion of pseudorandom error-correcting codes (PRCs) with the intent to watermark generative AI models. Informally, pseudorandom error-correcting codes are keyed coding schemes with the following three properties (see \cref{def: secret-keyPRC} and \cref{def: public-keyPRC} for details).
\begin{enumerate}
    \item
    Pseudorandomness: codewords are computationally indistinguishable from random for any algorithm which does not have the secret key.
    \item 
    Robustness: anyone with the secret key can decode corrupted codewords with overwhelming probability.
    \item 
    Soundness: any fixed $x \in \{ 0, 1\}^n$ has a negligible probability (where the probability is over the key generation algorithm) of being decoded to a message by the decoding algorithm.
\end{enumerate}

One of the beautiful insights of \cite{christ2024pseudorandom} is that PRCs can be used to watermark generative models if one simply reinterprets the generative model as a channel which corrupts the randomness it uses. Consider an abstracted polynomial time generative algorithm $\mathsf{Generate}$ that as part of its input takes in a random input seed $r \in \{ 0, 1\}^n$, and produces content $t \in \{ 0, 1\}^n$. We model an adversary trying to evade detection by a channel $\mathcal{E}': \{ 0, 1\}^n \rightarrow \{ 0, 1\}^n$ which corrupts the content $t$ into $\tilde{t}$. Furthermore, we assume that there exists an algorithm $\mathsf{Recover}$ which recovers an approximation $\tilde{r}$ of $r$ from $\tilde{t}$. The channel $\mathcal{E} = \mathsf{Recover} \circ \mathcal{E}' \circ \mathsf{Generate}$ then acts as a corrupting channel for the input seed $r$.

Say we wish to watermark the output of our generative model with some message $m$. Let $c$ be a PRC codeword for $m$. Notice that if we run $\mathsf{Generate}$ seeded with $c$, rather than truly random $r$, we obtain several desirable properties.
\begin{enumerate}
    \item Undetectability \cite{christ2024undetectable}: the pseudorandomness of PRC outputs guarantees that watermarked content is computationally indistinguishable from unwatermarked content. This guarantees that the quality of the outputs is not degraded by watermarking.

    \item Tamper resistance: applying our PRC's robust decoding algorithm to the tampered AI generated content $\mathcal{E}(c)$ lets us recover the watermark $m$. Therefore, the watermark is not removed by the tampering by $\mathcal{E}'$ to the generated content.

    \item Few false positives: the soundness property guarantees that for any fixed human generated text $z_1 \dots z_n$, with overwhelming probability, the decoding algorithm will not flag it as a corrupted codeword (and thus watermarked text).
\end{enumerate}

In this work, we are concerned with new constructions of pseudorandom codes. The assumptions required for the construction of pseudorandom codes in \cite{christ2024pseudorandom} are relatively strong (see \cref{subsec: revisit}), and were subsequently weakened in the case of secret-key PRCs to the existence of a local weak pseudorandom function family~\cite{golowich2024editdistancerobustwatermarks}.

We restrict ourselves to constructing zero-bit PRCs (the encoded message is always ``1'', see \cref{def: zero-bit}) which are robust to all channels which introduce errors at a rate of $1/2-\varepsilon$ (for constant $\varepsilon$). As shown in~\cite{christ2024pseudorandom}, such constructions can be  bootstrapped into constant rate PRCs (see \cref{def: rate}). Furthermore, \cite{christ2024pseudorandom, golowich2024editdistancerobustwatermarks} show how to bootstrap such constructions into codes which are robust to other types of errors than just substitution errors.

\subsection{Watermarking large language models}
We wish to emphasize that the framework of watermarking using PRCs is not restricted to any one type of generative AI model. However, to help make the philosophy of watermarking using PRCs more concrete, we review how \cite{christ2024pseudorandom} instantiate a PRC based scheme for watermarking large language models (LLMs).

Imagine an abstracted model of an LLM which works over the binary alphabet and always outputs text of length $n$. Concretely, consider an efficiently computable function $f: \{ 0, 1\}^* \times \{ 0, 1\}^* \rightarrow [0, 1]$ which takes in the prompt and the output text so far as the input and outputs the probability $p \in [0, 1]$ that the next token will be 1. The use of the binary alphabet in $f$ is without loss of generality since all tokens can be represented in binary. Text generation on a prompt $y \in \{ 0, 1\}^*$ works by iteratively sampling $z_{i} \leftarrow \text{Ber}(f(y, z_1 \dots z_{i-1}))$ for all $i \in [1, n]$. The final output of the LLM is then $z_1 \dots z_n$.

Let us now consider a different procedure to sample from the same distribution. We first sample $x_1 \dots x_n$, each independently from $\text{Ber}(1/2)$. To generate from the LLM on a prompt $y \in \{ 0, 1\}^*$, we iteratively sample $z_i$ for $i \in [1, n]$ as follows. Let $p_i = f(y, z_1, \dots, z_{i-1})$, if $p_i \leq 1/2$, sample $z_i$ from $\text{Ber}(2 p_i x_i)$, otherwise sample $z_i$ from $\text{Ber}(1-(1-x_i)(1-p_i))$. Note that since each $x_i$ is sampled uniformly from $\text{Ber}(1/2)$, $z_i$ is still distributed as $\text{Ber}(p_i)$, and therefore the output distribution of the LLM on a prompt $y$ remains unchanged from the previous example.

The key now is to create an LLM that samples $x_1 \dots x_n$ from a pseudorandom error-correcting code. We will call this new LLM the watermarking LLM. We assume that the original LLM is a polynomial time algorithm (formally, we need a family of LLMs parameterized by input length for the notion of a polynomial time algorithm to make sense, but we omit such details for the sake of exposition). Therefore, the \emph{pseudorandomness property} guarantees that the output distribution of the watermarking LLM is computationally indistinguishable from the case when $x_1 \dots x_n$ are sampled at random, which we just saw is the same as the original LLM output distribution. 

Furthermore, notice that if $0 < p_i <1$, then $z_i = x_i$ with probability greater than $1/2$. Therefore, if many $p_i$ are bounded away from one, the output of the watermarking LLM is relatively close to the codeword $x_1 \dots x_n$. The watermarking LLM takes the codeword $x_1 \dots x_n$ as one of its inputs and outputs $z_1 \dots z_n$, and in this way it functions as a corrupting channel. For sufficiently high entropy outputs, many $p_i$ are sufficiently close to $1/2$, therefore $z_1 \dots z_n$ is relatively close to $x_1 \dots x_n$, and anyone with the secret key can decode $z_1 \dots z_n$, thereby confirming that the output has been watermarked. Furthermore, the LLM output $z = z_1 \dots z_n$ is also robust to corruptions by an adversary trying to evade detection since $\Tilde{z}$ will still be decoded by someone with a secret key assuming that $\Delta(z, \Tilde{z})$ is small (which it will be if the adversary does not make significant changes to $z$). Therefore, watermarked and edited text corresponds to corrupted PRC codewords.

For a discussion of how to watermark LLM text using PRCs as well as a other application of PRCs (robust steganography), we refer the reader to \cite{christ2024pseudorandom}.
\subsection{Our results}
For the purpose of watermarking, our PRCs usually need to be robust to $p$-bounded channels (see \cref{def: p-bounded}). Informally, these are channels where an adversary can arbitrarily flip any $pn$ bits of a codeword.

We begin with a warmup in which we present a construction under the minimal cryptographic assumption of one-way functions (\cref{thm: warmup_thm}). We view this warmup as a way to build intuition about PRCs and what assumptions we use to construct them. We also view it as an interesting result telling us what parameters we should try to achieve in our constructions. The simple warmup PRC scheme shows that it is easy to build PRCs which are robust to any sub-constant noise rate over the binary alphabet or robust to any constant noise rate over increasing alphabet sizes. We therefore restrict our attention in the following sections to constructing PRCs which are robust to a constant noise rate, which surprisingly turns out to require much stronger assumptions and involved constructions.

\subsubsection{Planted hyperloop construction}
The planted hyperloop assumption, introduced in \cite{BKR}, asserts that a random 5-hypergraph is distinguishable from a random 5-hypergraph with a special $\Theta(\log n)$ size 3-hypergraph planted in it with advantage at most $o(1)$. \cite{BKR} show that if both the planted hyperloop assumption and the security of Goldreich's PRG \cite{Goldreich2011} instantiated with the predicate $P_5(x_1, \dots, x_5) = x_1 \oplus x_2 \oplus x_3 \oplus x_4 x_5$ hold, then public key cryptography exists. We show that that similar assumptions imply a type of public-key PRC.

\begin{theorem}[informal version of \cref{thm: hyperloop_bigthm}]
    Under the assumption used to construct public key cryptography in \cite{BKR} and $o(1)$-pseudorandomness of Goldreich's PRG instantiated with the $P_5(x_1, \dots, x_5) = x_1 \oplus x_2 \oplus x_3 \oplus x_4 x_5$ predicate, there exist public-key PRCs robust to $p$-bounded channels for constant $p<1/2$ and with $o(1)$ pseudorandomness against PPT adversaries.
\end{theorem}
Informally, by $\gamma$ pseudorandomness here, we mean that any PPT algorithm can distinguish a polynomial number of samples from random with at most $\gamma$ advantage.

This result mostly follows from observing that the \cite{BKR} construction exhibits some robustness to errors. We do require some care since \cite{BKR} are able to apply standard techniques to amplify security and correctness in their public key cryptography scheme, whereas such amplification techniques would break the robust decoding property needed in PRCs. 

There are at least two ways to interpret \cref{thm: hyperloop_bigthm}. The more obvious is simply the construction of public-key PRCs from studied cryptographic assumptions. However, one can also view it as suggesting that either the planted hyperloop assumption or the security of Goldreich's PRG with the $P_5$ predicate is a surprisingly strong assumption. In particular, since the only other known construction of public-key PRCs relies on fairly strong assumptions (see \cref{subsec: revisit}), this puts the assumptions of \cref{thm: hyperloop_bigthm} into a select group of assumptions implying public-key PRCs.

\subsubsection{Revisiting \cite{christ2024pseudorandom} and planted XOR assumption}
\label{subsec: revisit}
In \cref{sec: weakxor}, we revisit the assumptions under which \cite{christ2024pseudorandom} construct PRCs. Their construction is secure if either of the following hold
\begin{enumerate}
    \item
    The \xorassump{} and polynomial security of LPN with constant noise rate
    \item 
    $2^{O(\sqrt{n})}$ security of LPN
\end{enumerate}
We revisit the first of these assumptions. While polynomial security of LPN with constant noise rate is a very well established cryptography assumption, the \xorassump{} (introduced in \cite{sparse-k-sum}) is relatively new. It is therefore the most critical vulnerability in the \cite{christ2024pseudorandom} construction. Informally, the \xorassump{} says that a random matrix $G \in \{ 0, 1\}^{m \times n}$ modified so that $O(\log n)$ rows xor to $0^m$ is computationally indistinguishable from a truly random matrix. We therefore generalize and relax the assumption to what we call the weak \xorassump{} (see \cref{assump: sparse-xor}). Informally, the weak planted XOR assumption (with noise rate $\varepsilon$) says that a random matrix $G \in \{ 0, 1\}^{m \times n}$ modified so that $O(\log n)$ rows xor to a vector $v$ sampled from $\text{Ber}(m, \varepsilon)$ is computationally indistinguishable from a truly random matrix.

We observe that both LPN and the weak \xorassump{} have a noise rate parameter, $\eta$ and $\varepsilon$ respectively. 
We show that there is a wide range of points along the $\varepsilon, \eta$ parameter trade-off curve for which pseudorandom codes robust to a constant noise rate exist. 

\begin{restatable*}{theorem}{WeakXorFinal}
    \label{thm: weakxor_bigtm}
    For efficiently computable $m = \poly(n), t = O(\log n), \eta = o(1), \varepsilon = O(\log(m)/(\eta m))$ which are functions of $n$ and constant $p \in [0, 1/2)$, if $\textsf{XOR}_{m, t, \varepsilon}$ holds and $\mathsf{LPN}[\eta]$ holds, then there exists a $(1-\negl(n), 1-\negl(n), \negl(n))$-public-key PRC which is robust to all $p$-bounded channels and pseudorandom against all PPT adversaries.
\end{restatable*}

One can choose to read this result as saying more about the planted XOR assumption than the construction of PRCs. We will see that adding noise to the planted xor assumption seems to weaken it (for which we have some minor evidence \cref{subsec: evidence}) and interacts nicely with the LPN assumption. This seems to imply that the weak planted xor assumption may be the next natural variant of the planted xor assumption to study.

\subsubsection{Unconditional PRCs for space-bounded adversaries}
A natural and fundamental question in this area is whether we can prove the unconditional existence of PRCs (not based on cryptographic conjectures). To this end, we initiate the study of PRCs which are pseudorandom against polynomial time, space-bounded adversaries. Here we rely on the results showing that the problem of learning sparse parities (possibly with noise) is hard for space-bounded adversaries \cite{RazLearningRequiresGoodMemory, koltime-space, garg2021memorysample}.
\begin{theorem}[informal]
    There exists a zero-bit PRC with codeword length $O(n)$ that is robust to error rate $p$ for any constant $p<1/2$ and is \textbf{unconditionally} pseudorandom against adversaries which have $O(n^{3/2-\varepsilon})$ space and $\poly(n)$ time.
\end{theorem}

We emphasize that these results are for the one-way space model in which the adversary has $O(n^{3/2-\varepsilon})$ and tries to distinguish between a stream of random bits and a stream of codewords (and must write down any bits from the stream it wishes to recall later).

For context, in the space-bound cryptography setting, the best secret-key cryptography scheme where the communicating parties must store the entire length $O(n)$ ciphertext is only secure against $O(n^2)$ space adversaries. So while the gap between the power of the adversary and that of the players is admittedly small in our PRC (which is essentially a secret-key scheme with a robust decoding algorithm), the gap is still somewhat close to the best known for security against space bounded adversaries (where there is no requirement of robustness). To our knowledge, we are the first to study robust decoding schemes in the cryptographic space-bounded setting.

Unfortunately, our scheme is unlikely to be practical for watermarking generative AI as most generative models use more than $O(n^{3/2})$ auxiliary space (where $n$ is the size of the output of the generative model). One can therefore view this result as a first step towards practical unconditional PRCs. As the field of space-bounded cryptography progresses, one may hope we will eventually be able to construct PRCs which are pseudorandom against $O(n^{5})$ space and $\poly(n)$ time adversaries, which may indeed be practical for watermarking generative AI. Conversely, we believe that our scheme may already be useful for other use cases, such as robust steganography for particular types of steganographic channels (see \cite{christ2024pseudorandom} for details on robust steganography).

\subsection{Further directions}
\begin{enumerate}
    \item 
    The construction of public-key pseudorandom codes from unstructured assumptions is possibly the biggest question left open by this work. All known constructions rely on structured assumptions like hardness of the learning parity with noise problem or the planted hyperloop assumption. Even the construction of public-key PRCs from such a strong assumption as indistinguishability obfuscation would be new and interesting.
    \item 
    The weak \xorassump{} introduced in \cref{sec: weakxor} merits more cryptoanalytic study. Can we find more evidence that such an assumption is indeed weaker than the standard \xorassump{}?
    \item
    It may also be interesting to study from a theoretical perspective whether there exist a general set of generative model properties such that watermarking models with those properties using some explicit error correcting codes (from some class of constructions) rather than a PRC does not significantly degrade quality of model outputs.
\end{enumerate}

\subsection{Related work}
The idea of pseudorandom error-correcting codes was introduced in \cite{christ2024pseudorandom} with the intent of watermarking generative AI. They constructed a binary zero-bit encryption scheme robust to the bounded adversarial substitution channel and used that to construct a binary, constant rate PRC robust to both the bounded substitution channel and the random deletion channel. Followup work by Golowich and Moitra \cite{golowich2024editdistancerobustwatermarks} showed a construction of pseudorandom codes where the alphabet size grows polynomially in the output length of the code from a zero-bit PRC. They showed how to use such large alphabet pseudorandom codes to watermark LLM texts so that they are robust to bounded edit-distance channels (channels allowing insertions, substitutions, and deletions). Interestingly, their construction assumes the existence of a $O(\log n)$-local weak pseudorandom function family. This is quite similar to \cref{sec: hyperloop} which (among other things), assumes the security of Goldreich's PRG, which is $O(1)$-local. 

PRCs are perhaps most closely related to backdoored pseudorandom generators. Backdoored PRGs (first introduced in \cite{VVbackdooredPRG}) are pseudorandom generators where anyone with a secret key can distinguish PRG outputs from random. Zero-bit PRCs can just as well be thought of as backdoored PRGs where the mechanism to distinguish PRG outputs from random is robust to errors in its input.

The planted hyperloop construction of public-key cryptography \cite{BKR} is itself based on \cite{ApplebaumBarakWigdersonPKCFromDifferentAssump} and \cite{Goldreich2011}. These all belong to lines of work labeled expander-based cryptography which utilize or change the structure of expander graphs to build cryptographic primitives \cite{ExpanderBasedCryptoPlusNaturalProofs}.

\cref{sec: time-space} is based on \cite{RazLearningRequiresGoodMemory, koltime-space, garg2021memorysample}, which show that the problem of learning sparse parities with noise is hard for space-bounded algorithms. These results are intimately connected to the area of space-bounded cryptography. In space-bounded cryptography (introduced in \cite{MaurerUeli1992}), it is assumed all adversaries are space-bounded (have at most, say, $o(n^2)$ space, where $n$ is the message length). Unlike traditional cryptography, researchers have been able to prove unconditional results in the bounded storage setting \cite{KaliskiBurton1997, Ding2001, Dodis2023}.

    \subsection*{Acknowledgements}

    The authors would like to thank Yinuo Zhang for helpful discussion. They would also like to thank Sam Gunn and Noah Stephens-Davidowitz for reviewing early drafts of this work. We also wish to thank Miranda Christ for the observation that our warmup implies that secret-key PRCs with $\omega(1)$ alphabet size and robustness to constant rate errors is trivial to achieve.

\section{Preliminaries}
\label{sec:prelims}
\subsection{Notation}
We will use the notation ${{[n]} \choose k}$ to denote the set of all size $k$ subsets of $[n]$. We also often use the notation $x_{[a, b]}$ to denote bits $a$ through $b$ (inclusive) of the string $x$. We write $\text{Ber}(n, \eta)$ to denote the distribution $x_1 x_2 \dots x_n$ where each bit $x_i \in \{ 0, 1\}$ is sampled independently from Ber($\eta$).  We write BSC($p$) to denote the binary symmetric channel with crossover probability $p$. This is the channel where each bit is flipped with probability $p$ and remains the same with probability $1-p$. For $x, y \in \{ 0, 1\}^n$, $\Delta(x, y) = |\{ i : x_i \neq y_i \}|$ is the Hamming distance between $x$ and $y$. We also use $\mathcal{S}_{t, n} = \{ x \in \{ 0, 1\}^n : |x| = t \}$ to denote the Hamming sphere of dimension $n$ and radius $t$.

We will write $x_1, \dots, x_n \leftarrow \mathcal{D}$ to denote sampling $x_1, \dots, x_n$ each independently from a distribution $\mathcal{D}$ and also occasionally overload this notation by writing $x_1, \dots, x_n \leftarrow S$ to denote sampling $x_1, \dots, x_n$ each independently and uniformly from the set $S$.

If $a \in \{ 0, 1\}^n$ and $b \in \{ 0, 1\}^m$, then $ab \in \{ 0, 1\}^{n+m}$ denotes the concatenation of $a$ and $b$. For a matrix $G \in \{ 0, 1\}^{n \times m}$, $G_{i} \in \{ 0, 1\}^m$ is row $i$ of $G$.
\subsection{Probability and combinatorics}
\begin{definition}
    We say a string $a \in \{ 0,1 \}^n$ is $\delta$-biased if $|\{ i: a_i =0\} - \{ i : a_i = 1\}| \leq \delta n$.
\end{definition}

\begin{lemma}
    \label{lem: xoring_bias}
    Let $p_1, \dots, p_n \in [0, 1/2]$, if $X_i \sim \text{Bern}(p_i)$, then
    $$\underset{X_1, \dots, X_n}{\operatorname{Pr}}[X_1 \oplus \dots \oplus X_n = 0] = \frac{1}{2} \left( 1 + \prod_{i=1}^n (1-2p_i) \right) \ . $$
\end{lemma}

\begin{lemma}[Chernoff Bound \cite{Chernoff}]
    \label{lem: chernoff}
    Let $X_1, \dots, X_n$ be independent random variables, each distributed as Ber($p$). Let $\mu = np$, and $X = X_1 + \dots + X_n$. If $\delta \geq 0$, then
    $$\underset{X_1, \dots, X_n}{\operatorname{Pr}}[X \geq (1+\delta) \mu] \leq e^{-\delta^2 \mu/(2+\delta)} \ .$$
    If $0 < \delta < 1$, then
    $$\underset{X_1, \dots, X_n}{\operatorname{Pr}}[X \leq (1-\delta) \mu] \leq e^{-\delta^2 \mu/3} \ . $$
\end{lemma}

We will use the same insights as \cite{christ2024pseudorandom} to reduce the case of $p$-bounded adversarial channels to the case of the hypergeometric channel. For this we need the following lemma regarding the hypergeometric distribution. Let $\text{Hyp}(N, K, n)$ denote the distribution of the number of good elements chosen when choosing $n$ elements without replacement from a population of size $N$ which contains $K$ good elements.
\begin{lemma}[\cite{Hoeffding1994}]
    \label{lem: hypergeometric_chernoff}
    Let $X \sim \text{Hyp}(N, K, n)$ and $p = K/N$. Then for any $0<t<K/N$,
    $$\underset{}{\operatorname{Pr}}[X \geq (p+\varepsilon)n] \leq e^{-2 \varepsilon^2 n} \ . $$
\end{lemma}
\begin{restatable}{lemma}{xorBiasHyp}
    \label{lem: xor_bias_no_replacement}
    If $0 \leq t \leq m \leq n$, $X \sim \text{Hyp}(n, m, t)$, then
    \[ \frac{1}{2} + \frac{1}{2} \min_{\frac{m-t}{n} \leq p_i \leq \frac{m}{n-t}} \prod_{i=1}^t (1-2p_i)  \leq {\operatorname{Pr}}[\text{$X$ is even}] \leq \frac{1}{2} + \frac{1}{2} \max_{\frac{m-t}{n} \leq p_i \leq \frac{m}{n-t}} \prod_{i=1}^t (1-2p_i) \ .\]
\end{restatable}
\begin{restatable}{corollary}{xorBiasHypCor}
    \label{cor: xor_bias_no_replacement}
    If $0 \leq t \leq m \leq n$, $X \sim \text{Hyp}(n, m, t)$ and $p$ is a value maximizing $|1-2p|$ subject to $(m-t)/n \leq p \leq m/(n-t)$, then
    \[{\operatorname{Pr}}[\text{$X$ is even}] \leq \frac{1}{2} + \frac{1}{2} |1-2p|^t , \quad  \text{and}
    \quad    {\operatorname{Pr}}[\text{$X$ is odd}] \leq \frac{1}{2} + \frac{1}{2} |1-2p|^t \ .\]
\end{restatable}
\fullversion{
\noindent See \cref{sec: appendixB} for proofs of \cref{lem: xor_bias_no_replacement} and \cref{cor: xor_bias_no_replacement}.
}

\begin{lemma}
    \label{lem: birthday_bound}
    Let $X_1, \dots, X_Q$ be uniformly distributed over $[N]$.
    $$\underset{X_1, \dots, X_Q}{\operatorname{Pr}}[\exists i \neq j, X_i = X_j] \leq \frac{Q^2}{N} \ . $$
\end{lemma}

\begin{definition}
    The statistical distance (also known as the total variation distance) of two distribution $X$ and $Y$ on a finite domain $D$ is defined as
    $$\Delta(X, Y) = \frac{1}{2} \sum_{z \in D} \left| {\operatorname{Pr}}[X=z] - {\operatorname{Pr}}[Y=z] \right|$$
\end{definition}
We say two distributions $X$ and $Y$ are statistically indistinguishable if $\Delta(X, Y) = \negl(n)$.
\begin{fact}
    \label{fact: rejection_sample_SD}
    Let $A$ be a set and $B \subseteq A$. If $X$ is uniformly distributed over $A$, and $Y$ is uniformly distributed over $B$, then $\Delta(X, Y) = 1-|B|/|A|$.
\end{fact}

\subsection{Indistinguishability and LPN}
For a class of functions $\varepsilon$, we say two distribution ensembles $\{ D_n \}_{n \in \mathbb{N}}, \{ E_n \}_{n \in \mathbb{N}}$ are $\varepsilon$-indistinguishable if for any probabilistic polynomial time, non-uniform adversary $\mathcal{A}$, there exists a function $\varepsilon' \in \varepsilon$ such that
$$\left| \underset{
    x \leftarrow D_n}{\operatorname{Pr}}[\mathcal{A}(x) = 1] - 
    \underset{
    x \leftarrow E_n}{\operatorname{Pr}}[\mathcal{A}(x) = 1] \right| \leq \varepsilon'(n)$$
We say that $\{ D_n \}_{n \in \mathbb{N}}$ and $\{ E_n \}_{n \in \mathbb{N}}$ are computationally indistinguishable if they are $\negl(n)$-indistinguishable. 

The learning parity with noise (LPN) assumption is going to be critical in \cref{sec: weakxor}. For a linear code specified by a generator matrix $G$, an LPN sample is generated by sampling a random codeword $Gs$, and then adding some Bernoulli distributed noise $e$ to it. The LPN assumption says that an LPN sample is indistinguishable from random. Intuitively, the assumption says that noisy codewords from a random linear code are indistinguishable from random.
\begin{assumption}
    \label{assump: LPN}
    For $\eta: \mathbb{N} \rightarrow \mathbb{R}$ which is a function of $n$, the $\mathsf{LPN}[\eta]$ assumption states that for all $m = \poly(n)$ and all probabilistic $\poly(n)$ time algorithm $\mathcal{A}$,
    $$\Biggl| \underset{
    \begin{array}{c}G \leftarrow \F_2^{n \times m},\\ s \leftarrow \F_2^{m},\\ e \leftarrow \text{Ber}(n, \eta)\end{array}
    }{\operatorname{Pr}}[\mathcal{A}(G, Gs + e) = 1] - 
    \underset{
    \begin{array}{c}G \leftarrow \F_2^{n \times m},\\ u \leftarrow \F_2^{n}\end{array}
    }{\operatorname{Pr}}[\mathcal{A}(G, u) = 1] \Biggr| = \negl(n)$$
\end{assumption}
While \cref{assump: LPN} is stated for a single LPN sample, for a randomly sampled $G$, a polynomial number of LPN samples would still be computationally indistinguishable from random. This follows from a standard hybrid argument.

In our construction, we will actually rely on the following assumption where the secret $s$ is sampled from the same distribution as the noise. Lemma 2 of \cite{Fast_crypto_from_LPN} shows \cref{assump: LPN_sparse_secret} implied by \cref{assump: LPN}.
\begin{assumption}
    \label{assump: LPN_sparse_secret}
    For $\eta: \mathbb{N} \rightarrow \mathbb{R}$ which is a function of $n$, the $\mathsf{LPN}[\eta]$ assumption states that for all $m = \poly(n)$ and all probabilistic $\poly(n)$ time algorithm $\mathcal{A}$,
    $$\Biggl| \underset{
    \begin{array}{c}G \leftarrow \F_2^{n \times m},\\ s \leftarrow \text{Ber}(m, \eta),\\ e \leftarrow \text{Ber}(n, \eta)\end{array}
    }{\operatorname{Pr}}[\mathcal{A}(G, Gs + e) = 1] - 
    \underset{
    \begin{array}{c}G \leftarrow \F_2^{n \times m},\\ u \leftarrow \F_2^{n}\end{array}
    }{\operatorname{Pr}}[\mathcal{A}(G, u) = 1] \Biggr| = \negl(n)$$
\end{assumption}

\subsection{Pseudorandom Codes}
\begin{definition}[\cite{christ2024pseudorandom}]
    \label{def: p-bounded}
    We say that a length-preserving binary channel $\mathcal{E}: \{ 0, 1\}^* \rightarrow \{ 0, 1\}^*$ is $p$-bounded if for all $n \in \mathbb{N}$, $\underset{
    x \leftarrow \{ 0, 1\}^n}{\operatorname{Pr}}[|\mathcal{E}(x) \oplus x| > pn] \leq \negl(n)$.
\end{definition}

\fullversion{
\begin{definition}
    The $d$-hypergeometric channel $\mathcal{E}: \{ 0, 1\}^n \rightarrow \{ 0, 1\}^n$ is defined as the channel $\mathcal{E}$ which takes in $x$, samples $y \leftarrow \mathcal{S}_{d, n}$, and outputs $\mathcal{E}(x) = x \oplus y$.
\end{definition}
}

We now define secret and public key pseudorandom codes. 

\begin{definition}[Secret-key PRC \cite{christ2024pseudorandom}]
    \label{def: secret-keyPRC}
    Let $\Sigma$ be a fixed alphabet. An ($\alpha, \beta, \gamma$)-\emph{secret-key pseudorandom error-correcting code} (abbreviated as secret-key PRC) with robustness to a channel $\mathcal{E}: \Sigma^* \rightarrow \Sigma^*$ and pseudorandomness against a class of adversaries $\mathcal{C}$ is a triple of polynomial time randomized algorithms $(\mathsf{KeyGen}, \mathsf{Encode}, \mathsf{Decode})$ satisfying
    \begin{itemize}
        \item (Syntax) There exists functions $\ell, n, k: \mathbb{N} \rightarrow \mathbb{N}$ such that for all $\lambda \in \mathbb{N}$, $\mathsf{KeyGen}(1^\lambda) \in \{ 0,1\}^{\ell(\lambda)}$, $\mathsf{Encode}: \{ 1^\lambda \} \times \{ 0,1 \}^{\ell(\lambda)} \times \Sigma^{k(\lambda)} \rightarrow \Sigma^{n(\lambda)}$, and $\mathsf{Decode}: \{ 1^\lambda \} \times \{ 0,1 \}^{\ell(\lambda)} \times \Sigma^* \rightarrow \Sigma^{k(\lambda)} \cup \{ \bot \}$.

        \item (Error correction, or robustness) For any $\lambda \in \mathbb{N}$ and any message $\mathsf{m} \in \Sigma^{k(\lambda)}$,
        $$\underset{\mathsf{sk} \leftarrow \mathsf{KeyGen}(1^\lambda)}{\operatorname{Pr}}[\mathsf{Decode}(1^\lambda, \mathsf{sk}, \mathcal{E}(x)) = \mathsf{m} : x \leftarrow \mathsf{Encode}(1^\lambda, \mathsf{sk}, \mathsf{m})] \geq \alpha$$

        \item (Soundness) For any fixed $c \in \Sigma^*$,
        $$\underset{\mathsf{sk} \leftarrow \mathsf{KeyGen}(1^\lambda)}{\operatorname{Pr}}[\mathsf{Decode}(1^\lambda, \mathsf{sk}, c) = \bot] \geq \beta$$

        \item  (Pseudorandomness) For any adversary $\mathcal{A} \in \mathcal{C}$,
        $$\left| \underset{\mathsf{sk} \leftarrow \mathsf{KeyGen}(1^\lambda)}{\operatorname{Pr}}[\mathcal{A}^{\mathsf{Encode}(1^\lambda, \mathsf{sk}, \cdot)}(1^\lambda) = 1] - 
        \underset{\mathcal{U}}{\operatorname{Pr}}[\mathcal{A}^{\mathcal{U}}(1^\lambda) = 1] \right| \leq \gamma$$
        where $\mathcal{A}^{\mathcal{U}}$ means that the adversary has access to an oracle that, on any (even previously queried) input, outputs a freshly drawn uniform value from $\Sigma^{n(\lambda)}$. 
    \end{itemize}
\end{definition}

\begin{definition}[Public-key PRC \cite{christ2024pseudorandom}]
    \label{def: public-keyPRC}
    Let $\Sigma$ be a fixed alphabet. An $(\alpha, \beta, \gamma)$-\emph{public-key pseudorandom error-correcting code} (abbreviated as public-key PRC) with robustness to a channel $\mathcal{E}: \Sigma^* \rightarrow \Sigma^*$ and pseudorandomness against a class of adversaries $\mathcal{C}$ is a triple of polynomial time randomized algorithms $(\mathsf{KeyGen}, \mathsf{Encode}, \mathsf{Decode})$ satisfying
    \begin{itemize}
        \item (Syntax) There exists functions $\ell_{\mathsf{Dec}}, \ell_{\mathsf{Enc}}, n, k: \mathbb{N} \rightarrow \mathbb{N}$ such that for all $\lambda \in \mathbb{N}$, $\mathsf{KeyGen}(1^\lambda) \in \{ 0,1\}^{\ell_{\mathsf{Dec}}(\lambda)} \times \{ 0,1\}^{\ell_{\mathsf{Enc}}(\lambda)}$, $\mathsf{Encode}: \{ 1^\lambda \} \times \{ 0,1 \}^{\ell_{\mathsf{Enc}}(\lambda)} \times \Sigma^{k(\lambda)} \rightarrow \Sigma^{n(\lambda)}$, and $\mathsf{Decode}: \{ 1^\lambda \} \times \{ 0,1 \}^{\ell_{\mathsf{Dec}}(\lambda)} \times \Sigma^* \rightarrow \Sigma^{k(\lambda)} \cup \{ \bot \}$.

        \item (Error correction, or robustness) For any $\lambda \in \mathbb{N}$ and any message $\mathsf{m} \in \Sigma^{k(\lambda)}$,
        $$\underset{\mathsf{(sk, pk)} \leftarrow \mathsf{KeyGen}(1^\lambda)}{\operatorname{Pr}}[\mathsf{Decode}(1^\lambda, \mathsf{sk}, \mathcal{E}(x)) = \mathsf{m} : x \leftarrow \mathsf{Encode}(1^\lambda, \mathsf{pk}, \mathsf{m})] \geq \alpha$$

        \item (Soundness) For any fixed $c \in \Sigma^*$,
        $$\underset{\mathsf{(sk, pk)} \leftarrow \mathsf{KeyGen}(1^\lambda)}{\operatorname{Pr}}[\mathsf{Decode}(1^\lambda, \mathsf{sk}, c) = \bot] \geq \beta$$

        \item  (Pseudorandomness) For any adversary $\mathcal{A} \in \mathcal{C}$,
        $$\left| \underset{\mathsf{(sk, pk)} \leftarrow \mathsf{KeyGen}(1^\lambda)}{\operatorname{Pr}}[\mathcal{A}^{\mathsf{Encode}(1^\lambda, \mathsf{pk}, \cdot)}(1^\lambda, \mathsf{pk}) = 1] - 
        \underset{\mathcal{U}}{\operatorname{Pr}}[\mathcal{A}^{\mathcal{U}}(1^\lambda, \mathsf{pk}) = 1] \right| \leq \gamma$$
        where $\mathcal{A}^{\mathcal{U}}$ means that the adversary has access to an oracle that, on any (even previously queried) input, outputs a freshly drawn uniform value from $\Sigma^{n(\lambda)}$. 
    \end{itemize}
\end{definition}

\fullversion{
We will say that a $(\alpha, \beta, \gamma)$-public-key or $(\alpha, \beta, \gamma)$-secret-key PRC scheme $(\mathsf{KeyGen}, \mathsf{Encode}, \mathsf{Decode})$ is robust to a channel $\mathcal{E}$ and pseudorandom/secure against $\C$ if the definition \cref{def: public-keyPRC} or \cref{def: secret-keyPRC} respectively holds given $\mathcal{E}$ is instantiated as the channel and $\mathcal{C}$ is instantiated as the class of adversaries. We adopt this notation of $\mathcal{E}$ and $\mathcal{C}$ as implicit parameters in \cref{def: public-keyPRC} and \cref{def: secret-keyPRC} so as not to clutter the parameters but one can just as easily parameterize the definition by all relevant variables (e.g. $(\alpha, \beta, \gamma, \mathcal{E}, \C)$-public-key PRC).

We have expanded the definitions of secret-key PRC and public-key PRC from \cite{christ2024pseudorandom} by including the parameters $\alpha, \beta, \gamma$ in the definition and including the new implicit parameter $\C$ (which is necessary to formalize security against space-bounded adversaries \cref{sec: time-space}). If we take $\C$ to be all non-uniform, probabilistic, polynomial time (PPT) algorithms, and $\alpha = 1-\negl(\lambda), \beta = 1-\negl(\lambda), \gamma=\negl(\lambda)$, we recover the original definitions given in \cite{christ2024pseudorandom}. When $\C$ and $\mathcal{E}$ are clear from context, we will often say PRC to mean a $(1-\negl(\lambda), 1-\negl(\lambda), \negl(\lambda))$-public-key PRC or $(1-\negl(\lambda), 1-\negl(\lambda), \negl(\lambda))$-secret-key PRC.
}

\begin{definition}
    \label{def: rate}
    For \cref{def: public-keyPRC} and \cref{def: secret-keyPRC}, we define $k(\lambda)/n(\lambda)$ as the rate of a PRC.
\end{definition}

\begin{definition}
    \label{def: zero-bit}
    We say a PRC scheme is a zero-bit PRC scheme if the only message $\mathsf{m}$ that is ever encrypted is 1.
\end{definition}
The image of the decoding function of a zero-bit scheme should only be $\{ 1, \bot\}$ since we know that $0$ is never encoded by the PRC. Informally, a zero-bit PRC requires only that we distinguish corrupted PRC outputs from strings which are not PRC outputs. We will focus on zero-bit PRCs since when $\C$ is all PPT algorithms, \cite{christ2024pseudorandom} shows that the existence of a zero-bit secret-key or public-key PRC implies the existence of a secret-key or public-key PRC respectively which has essentially the same robustness as the original but a worse rate. See \cite{christ2024pseudorandom} for a formal statement.

Say we have a zero-bit encryption scheme where corrupted codewords are identified as such with probability $\alpha(\lambda)$, random words are identified as codewords with probability $\alpha(\lambda)-1/\poly(n)$, and any polynomial number of codewords are $\gamma$-indistinguishable from random. This is not quite a PRC since we do not have the soundness property. However, our next lemma shows that we can use such a scheme to construct a $(1-\negl(\lambda), 1-\negl(\lambda), \gamma)$ zero-bit PRC.
\begin{lemma}
    \label{lem: amplify_public}
    Suppose that there exist PPT algorithms $(\mathsf{KeyGen}, \mathsf{Encode}, \mathsf{Decode})$ such that
    \begin{enumerate}
        \item There exists functions $\ell_{\mathsf{Dec}}, \ell_{\mathsf{Enc}}, n, k: \mathbb{N} \rightarrow \mathbb{N}$ such that for all $\lambda \in \mathbb{N}$, $\mathsf{KeyGen}(1^\lambda) \in \{ 0,1\}^{\ell_{\mathsf{Dec}}(\lambda)} \times \{ 0,1\}^{\ell_{\mathsf{Enc}}(\lambda)}$, $\mathsf{Encode}: \{ 1^\lambda \} \times \{ 0,1 \}^{\ell_{\mathsf{Enc}}(\lambda)} \times \{ 1 \} \rightarrow \Sigma^{n(\lambda)}$, and $\mathsf{Decode}: \{ 1^\lambda \} \times \{ 0,1 \}^{\ell_{\mathsf{Dec}}(\lambda)} \times \Sigma^* \rightarrow \{ 1, \bot \}$.

        \item 
        $n(\lambda) = \poly(\lambda)$.
        
        \item
        \label{item: strong_robustness}
        For every $d \leq p \cdot n(\lambda)$, $d$-hypergeometric channel $\mathcal{E}$, and a $1-\negl(\lambda)$ fraction of keys $(\mathsf{sk}, \mathsf{pk}) \leftarrow \mathsf{KeyGen}(1^\lambda)$,
        $$\underset{\mathcal{E}}{\operatorname{Pr}}[\mathsf{Decode}(1^\lambda, \mathsf{sk}, \mathcal{E}(x)) = 1: x \leftarrow \mathsf{Encode}(1^\lambda, \mathsf{pk}, 1)] \geq \alpha(\lambda)$$
        where the randomness is over the randomness of the encoding algorithm and the errors of $\mathcal{E}$.
        
        \item
        \label{item: weak_soundness}
        There exists a $\delta(n) = 1/\poly(n)$ where $\alpha(\lambda)-\delta(\lambda) \geq 1/\poly(\lambda)$ such that for a $1-\negl(\lambda)$ fraction of keys $(\mathsf{pk}, \mathsf{sk}) \leftarrow \mathsf{KeyGen}(1^\lambda)$,
        $$\underset{x \leftarrow \{ 0, 1\}^n}{\operatorname{Pr}}[\mathsf{Decode}(1^\lambda, \mathsf{sk}, x) = 1] \leq \delta(\lambda) \ . $$
        \item
    
        For any $q = \poly(\lambda)$, $X_1, \dots, X_q \leftarrow \mathsf{Enc}(1^\lambda, \mathsf{pk}, 1)$ is $\gamma$-indistinguishable from $Y_1, \dots, Y_q \leftarrow \{ 0, 1\}^{n(\lambda)}$.
    \end{enumerate}
    Then for every constant $\varepsilon > 0$, there exists of a zero-bit $(1-\negl(\lambda), 1-\negl(\lambda), \gamma(\lambda))$-public-key PRC robust to any $(p-\varepsilon)$-bounded channel and pseudorandom against any PPT adversary.
\end{lemma}
\begin{proof}

\fullornot{
Let $\varepsilon > 0$ be an arbitrarily small constant. Say $\alpha(\lambda)-\delta(\lambda) \geq 1/\lambda^c$ for some constant $c$ and sufficiently large $n$, and let $t = \lambda^{100c}/\delta(\lambda)$. We now construct a $(1-\negl(\lambda), 1-\negl(\lambda), \gamma(\lambda))$-public-key PRC with key generation, encoding, and decoding functions $\mathsf{KeyGen'}, \mathsf{Encode'}, \mathsf{Decode'}$ respectively. 
    \begin{itemize}
        \item 
        $\mathsf{KeyGen'}(1^\lambda)$: Sample $(\mathsf{sk}, \mathsf{pk}) \leftarrow \mathsf{KeyGen}(1^\lambda)$, $z_1, \dots, z_t \leftarrow \{ 0, 1\}^{n(\lambda)}$, and a random permutation $\pi: [tn] \rightarrow [tn]$. Output $(\mathsf{sk'}= (\mathsf{sk}, z_1, \dots, z_t, \pi), \mathsf{pk'} = (\mathsf{pk}, z_1, \dots, z_t, \pi))$.
        \item 
        $\mathsf{Encode'}(1^\lambda, (\mathsf{pk}, z_1, \dots, z_t, \pi), 1)$: Let $a_i \leftarrow \mathsf{Encode}(1^\lambda, \mathsf{pk}, 1)$ for $i \in [1, t]$. Output $\pi((a_1 \oplus z_1) || \dots || (a_t \oplus z_t))$.
        \item 
        $\mathsf{Decode'}(1^\lambda, (\mathsf{sk}, z_1, \dots, z_t, \pi), x)$: Decompose $\pi^{-1}(x) \in \{ 0, 1\}^{nt}$ into $\Tilde{a}_1 || \dots || \Tilde{a}_t$ with $\Tilde{a}_i \in \{ 0, 1\}^{n(\lambda)}$ and let $a_i = \Tilde{a}_i \oplus z_i$. Let $w_i = 1$ if and only if $\mathsf{Decode}(1^\lambda, \mathsf{sk}, a_i) = 1$. Output $1$ if $\sum_{i=1}^t w_i \geq t \cdot \frac{\alpha(n)+\delta(n)}{2}$ and output $\bot$ otherwise.
    \end{itemize}
    Observe that all the algorithms needed to specify our new PRC scheme $\mathsf{PRC'} = (\mathsf{KeyGen'}, \mathsf{Encode'}, \mathsf{Decode'})$ inherit $\poly(\lambda)$ running times from $(\mathsf{KeyGen}, \mathsf{Encode}, \mathsf{Decode})$.
    
    Consider a codeword $x \in \{ 0, 1\}^{tn}$ generated by $\mathsf{Encode}'$ and let $x' = \mathcal{E}(x)$. Because of the shifts $z_1, \dots, z_t$ and permutation $\pi$ applied to the codeword, for the purpose of showing robustness, we can assume without loss of generality that that $\mathsf{Encode}'$ and $\mathsf{Decode}'$ do not apply the shift $z$ or the permutation $\pi$ and $\mathcal{E}$ samples $m$ values without replacement from some distribution with support on $[1, (p-\varepsilon) tn]$ and then distributes $m$ errors randomly into the $tn$ coordinates of the codeword $x$. One way to imagine sampling the errors $\mathcal{E}$ introduces is by sampling from some joint distribution $(P_1, \dots, P_t)$ where $P_i$ represents represents the number of errors to introduce into the $i$th block of $n$ bits, and then flipping $P_i$ bits of the $i$th block of $n$ bits randomly. Let us consider a fixed $m \in [1, (p-\varepsilon) tn]$, notice that the number of errors in the $i$th section of $n$ bits is exactly distributed as $\text{Hyp}(nt, m, n)$. By \cref{lem: hypergeometric_chernoff}, $P_i > pn$ with $\negl(n) = \negl(\lambda)$ probability. By union bound, there is a $1-\negl(\lambda)$ probability that $P_i \leq pn$ for all $i \in [1, t]$. Therefore, we can assume without loss of generality that the errors in each block of $n$ bits come from a $p$-hypergeometric channel.
    
    We now show robustness of this new scheme. Notice that during the computation of \\$\mathsf{Decode'}(1^\lambda, (\mathsf{sk}, z_1, \dots, z_t), x')$, each $w_i = 1$ independently with probability $\alpha(n)$ by \cref{item: strong_robustness} and since $P_i \leq pn$. Furthermore, $t \cdot ((\alpha(\lambda)+\delta(\lambda))/2) = t (\alpha(\lambda) - (\alpha(\lambda)-\delta(\lambda))/2) = t (\alpha(\lambda) - 1/(2\lambda^c))$. So, by \cref{lem: chernoff}, the probability that $\sum_{i=1}^t w_i < t \cdot \frac{\alpha(\lambda)+\delta(\lambda)}{2}$ is at most $e^{-(1/(2\lambda^c))^2 \alpha(\lambda) t/3} \leq e^{-(1/(2\lambda^c))^2 \delta(n) t/3} = \negl(\lambda)$ by our choice of $t$. Therefore, there is a $1-\negl(\lambda)$ chance that $\mathsf{Decode'}$ outputs $1$.

    We now show soundness of the new scheme. Consider a fixed word $x \in \{ 0, 1\}^{tn}$. If we sample $(\mathsf{sk}, z_1, \dots, z_t)$ and run $\mathsf{Decode'}(1^\lambda, (\mathsf{sk}, z_1, \dots, z_t), x) = \bot$, the probability $w_i = 0$ is exactly the probability that $\mathsf{Decode}(1^\lambda, s_i, \Tilde{a}_i \oplus z) = \bot$. This probability is $\delta(\lambda)$ by assumption since $\Tilde{a}_i \oplus z$ is a uniformly random value in $\{ 0, 1\}^n$. Therefore, by \cref{lem: chernoff}, the probability that $\sum_{i=1}^t w_i \geq  t \cdot \frac{\alpha(\lambda)+\delta(\lambda)}{2} = t \cdot (\delta(\lambda) + (\alpha(\lambda) - \delta(\lambda))/2) = t \cdot (\delta(\lambda) + 1/(2\lambda^c))$ is at most $e^{-(1/(2\lambda^c)^2) \delta(\lambda)t/3} = \negl(\lambda)$ (where the last equality follows from our choice of $t$). Therefore there is a $1-\negl(\lambda)$ chance that $\mathsf{Decode'}$ outputs $\bot$.

    Pseudorandomness of our new scheme follows immediately from the pseudorandomness of the old scheme since the new scheme simply consists of codewords from the old scheme concatenated together (with a shift $z$ and permutation $\pi$ applied on top).
    }
    {See full version.}
\end{proof}

\fullversion{
Just as in \cite{christ2024pseudorandom}, for the remainder of the paper, we will identify $n$ with the security parameter. This may be confusing since when $\lambda$ is the security parameter, we say $n(\lambda)$ is the length of the code. However, letting $n$ be the security parameter brings our notation in line with the works most closely related to our own \cite{BKR, christ2024pseudorandom, RazLearningRequiresGoodMemory}.
}

\section{A warmup}
\fullornot{
\label{sec: warmup}
Here we give informal an description of a zero-bit PRC schemes with is only robust to $o(n)$ errors and pseudorandom against any PPT adversary. Our reasons for presenting this warmup are twofold. First, we believe it provides intuition into what types of assumptions are good for constructing PRCs. Second, we find it interesting that building PRCs which are robust to any sub-constant error rate is surprisingly easy (\cref{thm: warmup_thm}) but building PRCs which are robust to a constant error rate seems to require much more involved constructions. Very similar PRF based constructions have already been described in \cite{christ2024pseudorandom, christ2024undetectable, kirchenbauer2023watermark, scott2022}. Though we note that those seem to achieve robustness to any $o(1/\log n)$ error rate, whereas the following scheme is robust to any $o(1)$ error rate.

We will examine a simple construction of PRCs, which, for \emph{any} $\tau(n) = \omega(1)$, is robust against $\text{BSC}(1/\tau(n))$. We choose $\text{BSC}(1/\tau(n))$ instead of $(1/\tau(n))$-bounded channels only for ease of presentation and one can achieve robustness to $p$-bounded channels by applying techniques similar to those used in \cref{lem: amplify_public}. 

\begin{theorem}
    \label{thm: warmup_thm}
    Let $p(n)$ be any $o(1)$ function. If one-way functions exist, then there exists a $(1-\negl(n), 1-\negl(n), \negl(n))$-private-key PRC scheme robust to $\text{BSC}(p(n))$ and pseudorandom against all PPT adversaries.
\end{theorem}

\begin{proof}
\fullornot{
Let $\tau(n) = 1/p(n)$. The existence of one-way functions implies the existence of a keyed pseudorandom function family $f_k: \{ 0, 1\}^{\sqrt{\tau(n)} \log(n)} \rightarrow \{ 0, 1\}^n$ using the GGM construction \cite{GGM}. Let $t = n^3$. The key generation algorithm selects the key $k \leftarrow \{ 0, 1\}^n$, the encoding algorithm samples $x_1, \dots, x_t \leftarrow \{ 0, 1\}^{\sqrt{\tau(n)} \log(n)}$ and outputs $x_1 || f_k(x_1) || \dots || x_t || f_k(x_t)$, and the decoding algorithm on an input $\Tilde{x}_1 || \Tilde{y}_1 || \dots || \Tilde{x}_t || \Tilde{y}_t$ for $\Tilde{x}_i \in \{ 0, 1\}^{\sqrt{\tau(n)} \log(n)}, \Tilde{y}_i \in \{ 0, 1\}^n$ outputs $1$ if $\Delta(f_k(\Tilde{x}_i), \Tilde{y}_i) \leq n/10$ for any $i \in [1, t]$ and $\bot$ otherwise. 

\textsc{Robustness.} Imagine subjecting a codeword $x_1 || f_k(x_1) || \dots x_t || f_k(x_t)$ from this scheme to $\text{BSC}(p)$ for $p = 1/\tau(n)$. As long as there exists a $i \in [t]$ such that the bits in $x_i$ remain unchanged and the bits of $f_k(x_i)$ suffer less than $n/10$ flips, the corrupted codeword will be decoded to zero. For a fixed $i$, the probability that the bits of $x_i$ remain unchanged is
$$\left( 1-p \right)^{\log(n) \sqrt{\tau(n)}} = \left( \left( 1-\frac{1}{\tau(n)} \right)^{\log(n) \tau(n)} \right)^{\frac{1}{\sqrt{\tau(n)}}} \geq \left( \frac{1}{n^2} \right)^{\frac{1}{\sqrt{\tau(n)}}}$$
for sufficiently large $n$. The probability that every $x_i$ is changed is at most
$$\left(1- \frac{1}{n^{\frac{2}{\sqrt{\tau(n)}}}} \right)^t = \negl(n)$$
by our choice of $t$.
Furthermore, for any fixed $i \in [t]$, the probability that more than $n/10$ corruptions occur in bits $f_k(x_i)$ is $\negl(n)$ by \cref{lem: chernoff}. By the union bound, there is a $1-\negl(n)$ chance that there exists some $i \in [t]$ such that $x_i$ is unchanged, and $f_k(x_i)$ suffers fewer than $n/10$ bit-flips.

\textsc{Soundness.} Suppose the scheme were not sound, and there existed a series fixed string (parameterized by $n$) $x_1^* || y_1^* || \dots || x_t^* || y_t^*$ where $x_i^* \in \{ 0, 1\}^{\sqrt{\tau(n)} \log(n)}, y_t \in \{ 0, 1\}^n$ which were decoded to $1$ with non-negligible probability (where the probability is over the key $k$). Notice that for a truly random function $f$, for any $i \in [1, t]$, $\Delta(f(x_i^*), y_i^*) \leq n/10$ with $\negl(n)$ probability (by \cref{lem: chernoff}). Now consider the polynomial time distinguisher $\mathcal{A}(1^n)$ which when given oracle access to a function $f': \{ 0, 1\}^{\sqrt{\tau(n)} \log(n)} \rightarrow \{ 0, 1\}^n$ outputs $1$ if and only if there exists an $i \in [t]$ such that $\Delta(x_i^*, f'(x_i^*)) \leq n/10$. Notice that by assumption, if $f'$ is a PRF, then $\mathcal{A}$ outputs $1$ with non-negligible probability, and if $f'$ is truly random, it outputs $1$ with negligible probability. Therefore, $\mathcal{A}$ distinguishes between $f_k$ and a truly random function with non-negligible probability. This contradicts the fact that $f_k$ is a PRF. We have arrived at a contradiction, so the scheme must be sound. 

\textsc{Pseudorandomness.}
We will show that the outputs of this scheme are pseudorandom against PPT adversaries by applying the hybrid lemma. Let $q= \poly(t) = \poly(n)$ and consider the following distributions:
\begin{enumerate}
    \item
    $(x_1||f_k(x_1)|| \dots || x_q||f_k(x_q))$ where $k, x_1, \dots, x_q \leftarrow \{ 0, 1\}^n$
    \item
    $(x_1||f(x_1)|| \dots || x_q||f(x_q))$ where $f$ is a random function and $x_1, \dots, x_q \leftarrow \{ 0, 1\}^n$
    \item
    $(x_1||f(x_1)|| \dots || x_q||f(x_q))$ where $f$ is a random function and $x_1, \dots, x_q \leftarrow \{ 0, 1\}^n$ conditioned on all $x_i$ being distinct
    \item
    $(x_1||y_1|| \dots || x_q||y_q)$ where $x_1, \dots, x_q, y_1, \dots, y_q \leftarrow \{ 0, 1\}^n$ conditioned on all $x_i$ being distinct
    \item
    $(x_1||y_1|| \dots || x_q||y_q)$ where $x_1, \dots, x_q, y_1, \dots, y_q \leftarrow \{ 0, 1\}^n$
\end{enumerate}
Distributions 1 and 2 are computationally indistinguishable by the pseudorandomness of $f_k$. Distributions 2 and 3 are statistically indistinguishable by \cref{fact: rejection_sample_SD} since conditioning on $x_i$ being distinct only eliminates $t^2/2^{\log(n) \sqrt{t(n)}} = \negl(n)$ fraction of possibilities (\cref{lem: birthday_bound}). Distribution 3 and 4 are the same. Distribution 4 and 5 are statistically indistinguishable by the same reasoning as showed distribution 2 and 3 are indistinguishable. Therefore, by the hybrid lemma, distributions 1 and 5 are indistinguishable. Thus, the scheme satisfies the pseudorandomness property.
}
{See full version.}
\end{proof}

We see that the minimal cryptographic assumption of one-way functions lets us achieve robustness to any $o(1)$ error rate. We also observe that the presented PRC can be turned into a PRC over a larger alphabet $|\Sigma| = \tau(n)$ which is robust to a constant error rate by simply identifying each block of $\log_2(\tau(n))$ bits with a symbol in $\Sigma$. This has the benefit that the probability that some block of $\sqrt{\tau(n)} \log(n)$ bits (in particular, one corresponding to $x_i$) remains unchanged when the codeword is subjected to a constant error rate channel goes up to $1-\negl(n)$. This resolves the bottleneck in our robustness proof above and allows us to achieve robustness to a constant error rate. Therefore, it is trivial to construct PRCs with superconstant alphabet size. This sets a baseline and tells us that for a scheme to be considered non-trivial, it must be robust to a constant error rate and have constant alphabet size. 

The critical weakness of the PRF construction is that for a codeword to be decoded correctly, all $\omega(\log n)$ bits of some $x_i$ must remain intact. One approach to fix this (used to construct secret-key PRCs in \cite{golowich2024editdistancerobustwatermarks}) is to use a local weak PRF family so that if $\Delta(x, x')$ being small implies that $\Delta(f(x), f'(x))$ is small. A different approach is to aim for a scheme in which the decoding algorithm looks at a small number of bits of the codeword. In particular, if the decoding algorithm only looks at $O(\log n)$ of the the codeword, we may be able to achieve robustness to a constant error rate. This is the intuition that guides all of our upcoming schemes. This is also the approach guiding the scheme proposed in \cite{christ2024pseudorandom}.

}
{See full version.}

\section{Planted hyperloop construction}
\label{sec: hyperloop}
In this section, we use will use the planted hyperloop assumption and the security of Goldreich's PRG instantiated with the $P_5$ predicate to construct a public-key PRC scheme.
\begin{theorem}
    \label{thm: hyperloop_bigthm}
    Let $\delta, m, \ell, t$ be the parameters specified in \cref{assump: planted_hyperloop} and $p$ be a constant in $[0, 1/2)$. Under \cref{assump: planted_hyperloop} and \cref{assumption: GoldreichsPRG}, there exists a $(1-\negl(n), 1-\negl(n), o(1))$-public-key PRC robust to any $p$-bounded channel and pseudorandom against all PPT adversaries.
\end{theorem}

\subsection{The assumptions}
We begin by reviewing the assumptions used by \cite{BKR} to construct public-key cryptography. All hypergraphs are assumed to have ordered hyperedges (a hyperedge is an ordered tuple of vertices rather than a set of vertices).
\begin{definition}
    A hyperloop is a 3-hypergraph where each vertex has degree two and we define the size of a hyperloop as the number of hyperedges it contains.
\end{definition}
 The construction of \cite{BKR} plants $t = 2^{\Theta(\ell)}$ hyperloops $S_1, \dots, S_t$ of size $\ell = O(\log n)$ into a random hypergraph to create a 5-hypergraph $H$. The secret key is the set of $t$ hyperloops and the public key is $H$. We now formally present this construction.

\begin{construction}[\cite{BKR}]
    \label{cons: planted_hyperloop}
    Let $L_0$ be a fixed hyperloop of size $\ell = O(\log n)$. $H$  is sampled as follows. Let:
    \begin{enumerate}
    \itemsep=0ex
        \item
        $L$ be the union of $t = 2^{\Theta(\ell)}$ vertex-disjoint copies of $L_0$,
        \item 
        $Q$ be a random 3-hypergraph with $n$ vertices and hyperedge probability $O(n^{-3/2-\delta})$,
        \item 
        $P = Q \cup L$ where $L$ is planted on a random subset of the vertices of $Q$,
        \item
        If $P$ has more than $n^{3/2-\delta}$ hyperedges, output $H = \bot$. Otherwise, $P'$ is obtained by adding random 3-hyperedges to $P$ until it has $m = n^{3/2-\delta}$ hyperedges.
        \item 
        $H$ is obtained by randomly adding 2 vertices to each hyperedge in $P'$ (where those 2 vertices will be the last two in the ordered hyperedge)
    \end{enumerate}
    The public key is the 5-hypergraph $H$ and the secret key is $S_1, \dots, S_t$ where $S_i \subseteq \{ 1, \dots, m \}$ are they hyperedges corresponding to the $i^{th}$ planted copy of $L_0$. 
\end{construction}
\begin{figure}
    \centering
    \resizebox{.9\linewidth}{!}{\includegraphics{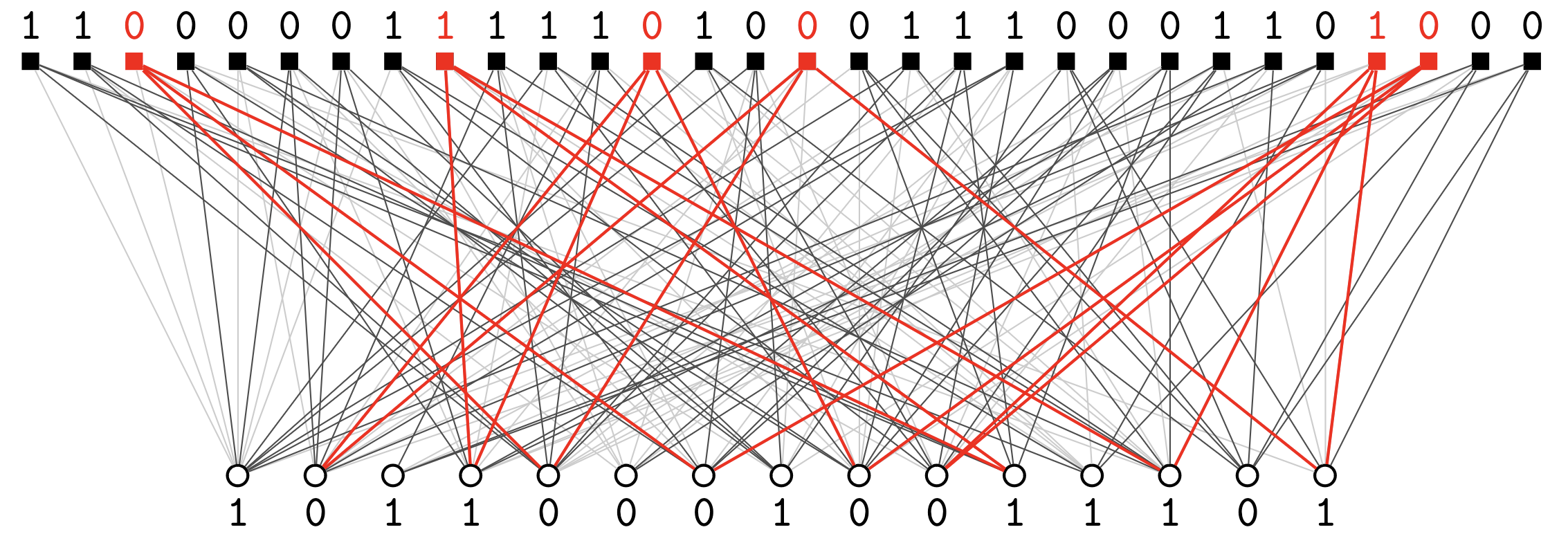}}
    \caption{A public key and PRG output with a single planted hyperloop $L_0$. The secret key is marked in red (from \cite{BKR}).}
    \label{fig:planted_hpyerloop}
\end{figure}

For consistency, we define $F_{\bot}(x) = 0^{(n^{-3/2-\delta})}$. We refer to any hypergraph generated using \cref{cons: planted_hyperloop} as a planted hyperloop graph. This leads us to our first assumption.

\begin{assumption}
    \label{assump: planted_hyperloop}
    For a sufficiently small constant $\delta$, $\ell = 0.36 \log n$, and $t = n^{0.75 - \delta}$, $P$ and $Q$ are $o(1)$-indistinguishable in $n^{O(1)}$ time.
\end{assumption}
\cref{assump: planted_hyperloop} is the planted hyperloop assumption of \cite{BKR} with the exception that it assumes $o(1)$-indistinguishability rather than $(1-\Omega(1))$-indistinguishability.
\fullversion{
However, in the following statement, \cite{BKR} indicates that $o(1)$-indistinguishability is also a fair assumption:
\begin{displayquote}
    Our security argument applies to distinguishers of any \emph{constant} advantage.
\end{displayquote}
The fact that the security arguments of \cite{BKR} apply to distinguishers of \emph{any} constant advantage rather than for some fixed constant advantage which is less than one leads us to believe that the assumption of $o(1)$-indistinguishability in \cref{assump: planted_hyperloop} is indeed fair. If one wishes to use the more conservative assumption of $(1-\Omega(1))$-indistinguishability for \cref{assump: planted_hyperloop}, \cref{thm: hyperloop_bigthm} holds with $(1-\Omega(1))$-indistinguishability rather than $o(1)$-indistinguishability.

We note $(1-\Omega(1))$-indistinguishability is a perfectly fine assumption for \cite{BKR} since they are able to amplify this to $\negl(n)$-indistinguishability using standard amplification techniques. However, applying such amplification techniques would mean that our PRC would not be robust to a constant error rate anymore. So we must make do with worse than $\negl(n)$-indistinguishability.

\begin{remark}
    \label{remark: quasi-polynomial-security}
    Since $\ell = O(\log n)$, a brute-force $2^{O(\log^2(n))}$ time algorithm can distinguish $P$ from $Q$ by searching for the implanted hyperloop. This will imply that \cref{cons: enc_hyperloop} will not be secure against $2^{O(\log^2(n))}$ time adversaries. Indeed, to our knowledge, all PRC are vulnerable to quasi-polynomial time adversaries.
\end{remark}
}

Hypergraphs with certain parameters (including planted hyperloop hypergraphs) can be used as PRGs. To show how, we review Goldreich's PRG. Fix the predicate $P_5(x_1, \dots, x_5) = x_1 \oplus x_2 \oplus x_3 \oplus x_4 x_5$. For an $n$ vertex, $m$ hyperedge, 5-hypergraph $H$, we define the PRG $F_H: \{ 0, 1\}^n \rightarrow \{ 0, 1\}^m$ as follows. On an input $x$, the bits of $x$ are projected onto the vertices of $H$, and bit $i$ of $F_H(x)$ is given by applying $P_5$ to the labeling of the vertices of hyperedge $i$. 

\cref{fig:planted_hpyerloop} gives a way to visualize Goldreich's PRG. We interpret our hypergraph $H$ as a bipartite graph $B$ where the input vertices of $B$ represent vertices of $H$, the output vertices of $B$ represent the hyperedges of $H$, and edge $(a, b) \in B$ if and only if vertex $a$ is contained in hyperedge $b$ in $H$. See \cref{fig:planted_hpyerloop} for the bipartite graph visualization with an example of the computation of $F_H$. Our second assumption is the security of Goldreich's PRG instantiated with the $P_5$ predicate.

\begin{assumption}
    \label{assumption: GoldreichsPRG}
    For every $\delta$, $m = n^{1.5 - \delta}$, and $s = \poly(n)$, random $Q$ belonging to set of 5-hypergraphs on $n$ verticies with $m$ hyperedges, and random $x_1, \dots, x_s \in \{ 0, 1\}^n$, $y_1, \dots, y_s \in \{ 0, 1\}^m$, $(Q, F_Q(x_1), \dots, F_Q(x_{s}))$ and $(Q, y_1, \dots, y_s)$ are $o(1)$-indistinguishable in $n^{O(1)}$ time. 
\end{assumption}
We now justify \cref{assumption: GoldreichsPRG}. Notice that it is too much to hope for $\negl(n)$-indistinguishability in \cref{assumption: GoldreichsPRG}. To see why, notice that there is a $\Omega(1/n^5)$ chance that the first two hyperedges in $Q$ contain the exact same vertices in the same order, which would cause the first bit of the output of $F_Q$ to always equal the second bit of the output. Such a function $F_Q$ is clearly not a PRG. 

\fullversion{
The authors of \cite{BKR} use the weaker assumption that for  every $\delta$, $m = n^{1.5 - \delta}$, random $Q, x \in \{ 0, 1\}^n, y \in \{ 0, 1\}^m$, that $(Q, F_Q(x))$ and $(Q, y)$ are $o(1)$-indistinguishable in $n^{O(1)}$ time. This differs from \cref{assumption: GoldreichsPRG} in that it only guarantees $o(1)$-indistinguishability for one sample from the PRG. However,} \cref{assumption: GoldreichsPRG} is in line with standard assumption about Goldreich's PRG \cite{Min_Complexity_Goldreich, CouteauConcreteGoldreich}.

\subsection{The construction}
\begin{construction}[{Hyperloop Construction, \textsf{Hyperloop}[$\delta, m, \ell, t$]}]
\label{cons: enc_hyperloop}
Let $\delta, m, \ell, t$ be efficiently computable functions of the security parameter $n$.
\begin{itemize}
    \item
    \textsf{KeyGen}$(1^n)$: Sample $H$ and $S$ as in \cref{cons: planted_hyperloop} conditioned on the last two vertices of each hyperedge in $S_1$ being pairwise disjoint. output $(\textsf{sk}=S, \textsf{pk}=H)$.
    \item
    \textsf{Encode}$(1^n, H, 1)$: Sample $u \leftarrow \{ 0, 1\}^{n}$ and output $F_H(u)$.
    \item 
    \textsf{Decode}$(1^n, S_1, x)$: Compute $w = \bigoplus_{j \in S_1} x_j$. If $w = 0$, output $1$, otherwise output $\bot$.
\end{itemize}
\end{construction}

\subsection{Robustness}
We first review a basic fact about decoding in planted hyperloop graphs when the output is not subjected to errors and then show that this decoding mechanism is robust to errors.
\begin{lemma}[{Claim 1 in \cite{BKR}}]
    \label{lem: bias_hyperloop}
    If $H, S$ come from \cref{cons: planted_hyperloop} where the hyperedges of $S_1$ are disjoint, if we sample $x$ uniformly at random from $\{ 0, 1\}^m$ and let $y = F_H(x)$, then $w = \bigoplus_{j \in S_1} y_j$ has bias $2^{-\ell}$ towards being $0$.
\end{lemma}
\begin{proof}
    \fullornot{
    We use the notation $x_{j_i}$ to denote the value of the $i$th vertex of hyperedge $j$ when $x$ is projected onto the vertices of $H$. Since all vertices in $S_1$ have degree two, $\bigoplus_{j \in s_1} (x_{j_1} \oplus x_{j_2} \oplus x_{j_3}) = 0$. So, $\bigoplus_{j \in S_1} y_j = \bigoplus_{j \in S_1} (x_{j_1} \oplus x_{j_2} \oplus x_{j_3} \oplus x_{j_4} x_{j_5}) = \bigoplus_{j \in S_1} x_{j_4} x_{j_5}$. Notice that each term $x_{j_4} x_{j_5}$ in the xor is 1 with probability $1/4$. Furthermore, since the last two vertices of each hyperedge in $S_1$ are disjoint, the terms $x_{j_4} x_{j_5}$ are independent for each $j$. Therefore, the bias of the xor of these biased independent terms, $\bigoplus_{j \in S_1} y_j$, has bias $2^{-\ell}$.
    }
    {See full version.}
\end{proof}

We now use this to prove that the output of Goldreich's PRG instantiated with planted hyperloop graphs is indeed a robust to errors.
\begin{lemma}
    \label{lem: hyperloop_robustness}
    Let $\delta, m, \ell, t$ be the parameters specified in \cref{assump: planted_hyperloop} and $p$ be any constant in $[0, 1/2)$. There exists some polynomial $p(n)$ such that for \cref{cons: enc_hyperloop}, for all $d \leq pm$, for all keys $(\mathsf{sk}, \mathsf{pk}) \leftarrow \mathsf{KeyGen}(1^n)$,
        $$\underset{\mathcal{E}}{\operatorname{Pr}}[\mathsf{Decode}(1^n, \mathsf{sk}, \mathcal{E}(x)) = 1: x \leftarrow \mathsf{Encode}(1^n, \mathsf{pk}, 1)] \geq \frac{1}{2} + \frac{1}{p(n)} \ .$$
        Here $\mathcal{E}$ is the $d$-hypergeometric channel and the randomness is over the randomness of the encoding algorithm and the errors of $\mathcal{E}$.
\end{lemma}
\begin{proof}
    \fullornot{
    Let $p'(n)$ be any polynomial such that the following holds:
    \[\frac{1}{2} \left( 1-2 \frac{pm}{m-\ell} \right)^\ell \geq \frac{1}{2} + 1/p'(n) \ .\] We will show that for all keys, over the randomness in the encoding function and $\mathcal{E}$, that a codeword from the PRC has a decent chance of being decoded to 1. Fix the key public and private key $H, S$. Formally, we wish to show that $$\underset{u, \mathcal{E}}{\operatorname{Pr}}[\textsf{Decode}(1^n, S_1, \mathcal{E}(F_H(u))) = 1] = \frac{1}{2} + \frac{1}{p(n)}$$
    Since $\mathcal{E}$ is the $d$-hypergeometric channel, we can model channel $\mathcal{E}$ as an error vector $e \leftarrow \mathcal{S}_{d, m}$. So the above statement is equivalent to
    $$\underset{u, e}{\operatorname{Pr}}[\textsf{Decode}(1^n, S_1, F_H(u) \oplus e) = 1] = \frac{1}{2} + \frac{1}{p(n)}$$
    Let us determine the probability that $w_i$ is zero for a fixed $H, S$. By \cref{lem: bias_hyperloop} and \cref{lem: xor_bias_no_replacement} respectively,
    $$\underset{u}{\operatorname{Pr}}[\oplus_{j \in S_1} (F_H(u)_j) = 0] \geq \frac{1}{2} + \frac{1}{n}$$
    $$\underset{e}{\operatorname{Pr}}[\oplus_{j \in S_1} e_{j} = 0] \geq \frac{1}{2} + \frac{1}{2} \left( 1-2 \frac{pm}{m-\ell} \right)^\ell \geq \frac{1}{2} + \frac{1}{p'(n)} \ .$$
    Since $F_H(u_i)$ and $e_i$ are independent, the events $\oplus_{j \in S_1} (F_H(u_i)_j) = 0$ and $\oplus_{j \in S_1} e_{ij} = 0$ are independent, and therefore, by \cref{lem: xoring_bias} we have
    $$\underset{u, e}{\operatorname{Pr}}[w = 0] \geq \frac{1}{2} + \frac{1}{n \cdot p'(n)}\ . \qedhere $$
    }
    {See full version.}
\end{proof}

\subsection{A form of soundness}
\begin{lemma}

    \label{lem: hyperloop_pseudosoundness}
    Let $\delta, m, \ell, t$ be the parameters specified in \cref{assump: planted_hyperloop}. For \cref{cons: enc_hyperloop}, for any key pair $(\mathsf{pk}, \mathsf{sk}) \leftarrow \mathsf{KeyGen}(1^n)$,
        $$\underset{x \leftarrow \{ 0, 1\}^n}{\operatorname{Pr}}[\mathsf{Decode}(1^n, \mathsf{sk}, x) = 1] = 1/2 \ . $$
\end{lemma}
\begin{proof}
    \fullornot{
    Recall that decoding in \cref{cons: enc_hyperloop} outputs $\bot$ if some fixed $O(\log n)$ bits (depending on $\mathsf{sk}$) of $x$ xor to $1$. For a random string $x$, this happens with probability exactly $1/2$.
    }
    {See full version.}
\end{proof}

\subsection{Pseudorandomness}
\begin{lemma}
    \label{lem: H_ind}
    Let $\delta, m, \ell, t$ be the parameters specified in \cref{assump: planted_hyperloop}. Under \cref{assump: planted_hyperloop}, $H$ is $o(1)$-indistinguishable from a random 5-hypergraph with $n$ vertices and $m = n^{1.5-\delta}$ edges.
\end{lemma}
\begin{proof}
    \fullornot{
    Let $\mathcal{D}$ be uniform distribution on 5-hypergraphs with $n$ vertices and $m = n^{1.5-\delta}$ edges and $\overline{P}$ be the $P$ rejection sampled to ensure $P$ has at most $n^{3/2-\delta}$ hyperedges. Let $T(G)$ be the function that outputs $\bot$ if $G$ has more than $n^{3/2-\delta}$ hyperedges, otherwise adds hyperedges to $G$ until it has $n^{3/2-\delta}$ hyperedges and then adds two verticies at random to each hyperedge. Recall that, under \cref{assump: planted_hyperloop}, $P$ and $Q$ are $o(1)$-indistinguishable. Notice that $T$ is efficiently computable. Therefore $H = T(P)$ and $T(Q)$ are $o(1)$-indistinguishable. Let $\overline{Q}$ be the distribution obtained by rejection sampling $Q$ to ensure that the $\overline{Q}$ has at most $n^{3/2-\delta}$ hyperedges. Since there is only a $\negl(n)$ chance that $Q$ has more than $n^{3/2-\delta}$ hyperedges, $Q \approx_c \overline{Q}$. Therefore, $T(Q) \approx_c T(\overline{Q})$. It is also not hard to see that $T(\overline{Q}) = \mathcal{D}$. Therefore, $\mathcal{D} = T(\overline{Q}) \approx_c T(Q)$ is $o(1)$-indistinguishable from $H = T(P)$. This clearly implies that $H$ is indistinguishable from a random 5-hypergraph with $n$ vertices and $m = n^{1.5-\delta}$ edges.
    }
    {See full version.}
\end{proof}

\begin{lemma}
    \label{lem: hyperloop_pseudorandomness}
    Let $\delta, m, \ell, t$ be the parameters specified in \cref{assump: planted_hyperloop}. Under \cref{assump: planted_hyperloop} and \cref{assumption: GoldreichsPRG}, the outputs of \cref{cons: enc_hyperloop} are $o(1)$-indistinguishable from random by any PPT adversary.
\end{lemma}
\begin{proof}
    \fullornot{
    Say that the distinguishing algorithm gets $s = \poly(n)$ samples. We will now define several distributions. Let $H$ be the distributions defined in \cref{cons: planted_hyperloop}, $H'$ be the distribution defined for $H$ in \cref{cons: enc_hyperloop}, $G$ be the uniform distribution over 5-hypergraphs with $n$ vertices and $m$ hyperedges, $X_1, \dots, X_s$ all be uniform on $\{ 0, 1\}^n$, and $Y_1, \dots, Y_s$ all be uniform on $\{ 0, 1\}^m$. For any graph $g$, let $g(x)$ denote the output of Goldreich's PRG instantiated with the $P_5$ predicate and graph $g$ on input $x$. By \cref{fact: rejection_sample_SD} and the fact that the $\mathsf{KeyGen}$ algorithm in \cref{cons: enc_hyperloop} rejects a $o(1)$ fraction of keys, the statistical distance between $H$ and $H'$ is $o(1)$. Therefore, $(H, H(X_1), \dots, H(X_s))$ and $(H', H'(X_1), \dots, H'(X_s))$ are distinguishable with advantage at most $o(1)$. By \cref{lem: H_ind}, $(H, H(X_1), \dots, H(X_s))$ is distinguishable from $(G, G(X_1), \dots, G(X_s))$ with $o(1)$ advantage. By \cref{assumption: GoldreichsPRG}, $(G, G(X_1), \dots, G(X_s))$ is distinguishable from $(G, Y_1, \dots, Y_m)$ with at most $1-o(1)$ advantage. These three facts along with the hybrid lemma tell us that $(H', H'(X_1), \dots, H'(X_s))$ is distinguishable from $(G, Y_1, \dots, Y_s)$ with advantage at most $o(1)$.
    }
    {See full version.}
\end{proof}

\subsection{Putting it all together}
\begin{proof}[Proof of \cref{thm: hyperloop_bigthm}]
    \fullornot{
    Observe that the length of the codewords generated by \cref{cons: enc_hyperloop} is $\poly(n)$. Furthermore, \cref{lem: hyperloop_robustness}, \cref{lem: hyperloop_pseudosoundness}, and \cref{lem: hyperloop_pseudorandomness} give us the remaining necessary preconditions to apply \cref{lem: amplify_public} which gives $(1-\negl(n), 1-\negl(n), o(1))$-public-key PRC.
    }
    {See full version.}
\end{proof}

The $o(1)$ pseudorandomness in \cref{thm: hyperloop_bigthm} is not ideal, but for the purpose of watermarking rather than security, seems tolerable. There is only a $o(1)$ probability that watermarking will ever have any noticeable effect (to someone who does not have the secret key).



\section{The weak planted XOR construction}
\label{sec: weakxor}

Christ and Gunn~\cite{christ2024pseudorandom} gave a scheme which is secure if both the \xorassump{} and polynomial hardness of LPN with constant noise rate hold. \footnote{For a particular setting of parameters, it is also secure if LPN with constant noise rate is $2^{O(\sqrt{n})}$ hard.} While polynomial hardness of LPN with constant noise rate is a well believed assumption, the \xorassump{} is a non-standard and relatively unstudied assumption. Therefore, the conjunction of these two assumptions is quite a strong assumption. Therefore, it seems plausible that a scheme based on a strengthened LPN assumption and a weakened \xorassump{} (denoted $\mathsf{XOR}_{m, t, \varepsilon}$) is more secure than the one presented in \cite{christ2024pseudorandom}, and this is what we show in this section.

\WeakXorFinal

\subsection{The assumption}
Let us define $\mathcal{D}_0(m, n)$ as the uniform distribution over $\{ 0, 1\}^{n \times m}$ and we will now define the distribution $\mathcal{D}_1(n, m, t, \varepsilon)$ which corresponds to the distribution of matrices where we strategically implant a low weight vector in the row space.
\begin{construction}[Generalization of \cite{sparse-k-sum}]
    \label{const: sparse_xor_sample}
    We define the distribution $\mathcal{D}_1(n, m, t, \varepsilon)$
    \begin{enumerate}
        \item
        Sample $G \leftarrow \{ 0, 1\}^{n \times m}$,
        \item
        Choose a random tuple $(a_1, \dots, a_t) \subseteq [n]^t$ such that $i \neq j$ implies $a_i \neq a_j$,
        \item
        Let $u = G_{a_1} \oplus \dots \oplus G_{a_{t-1}}$, $v \leftarrow \text{Ber}(m, \varepsilon)$, and update $G_{a_t}$ to $u + v$
        \item
        Output $(G, s)$ where $s \in \{ 0, 1\}^n$ is the $t$ sparse indicator vector for $(a_1, \dots, a_t)$.
        \end{enumerate}
\end{construction}

We are now ready to introduce the (weak) \xorassump{}. 
\begin{assumption}
    \label{assump: sparse-xor}
    For $m, t: \mathbb{N} \rightarrow \mathbb{N}$ and $\varepsilon: \mathbb{N} \rightarrow [0, 1/2]$ which are efficiently computable functions of $n$, the $\mathsf{XOR}_{m, t, \varepsilon}$ assumption states that for every probabilistic polynomial-time adversary $\mathcal{A}$,
    $$\left| \underset{G \leftarrow \mathcal{D}_0(n, m)}{\operatorname{Pr}}[\mathcal{A}(G) = 1] - \underset{(G,s) \leftarrow \mathcal{D}_1(n, m, t, \varepsilon)}{\operatorname{Pr}}[\mathcal{A}(G) = 1] \right| = \negl(n)$$
\end{assumption}
What is refereed to as the \xorassump{} in \cite{christ2024pseudorandom} is simply $\mathsf{XOR}_{m, O(\log n), 0}$.
As one of their major contributions, the authors of \cite{christ2024pseudorandom} give a PRC scheme which is secure if (i) \cref{assump: sparse-xor} with $\varepsilon = 0$, and $m = n^{1-\Omega(1)}, t = \Theta(\log n)$ is true, and (ii) constant noise rate LPN is hard. In this section, we show that such a scheme can be based on a more expansive set of assumptions. Informally, we will show that if for any $m = \poly(n)$, $\mathsf{XOR}_{m, \Theta(\log n), O(\log(m)/(m \eta))}$ holds and $\mathsf{LPN}[\eta]$ holds, then pseudorandom codes exist. For concreteness, one may wish to read this section with the parameter regime $\eta = 1/\sqrt{n}$ in mind since $\mathsf{LPN}[1/\sqrt{n}]$ is a well believed assumption and the weakest LPN assumption known to imply public-key cryptography \cite{alekhnovichMoreAverageCase2003}.

\subsection{Evidence $\mathsf{XOR}_{m, t, \varepsilon}$ is a weaker assumption than $\mathsf{XOR}_{m, t, 0}$}
\label{subsec: evidence}
Before proceeding with our PRC construction, we give two pieces of evidence that $\mathsf{XOR}_{m, t, \varepsilon}$ is indeed a weaker assumption than $\mathsf{XOR}_{m, t, 0}$. The first is a reduction which shows that $\mathsf{XOR}_{m, t, 0}$ implies $\mathsf{XOR}_{m, t, \varepsilon}$.

\begin{theorem}
    For any $m, t: \mathbb{N} \rightarrow \mathbb{N}$ and $\varepsilon: \mathbb{N} \rightarrow [0, 1/2]$ which are efficiently computable functions of $n$, if the $\mathsf{XOR}_{m, t, 0}$ assumption holds and $\text{Ber}(n, \varepsilon)$ is efficiently sampable, then the $\mathsf{XOR}_{m, t, \varepsilon}$ assumption holds.
\end{theorem}
\begin{proof}
    \fullornot{
    We will show the contrapositive statement that if $\mathsf{XOR}_{m, t, \varepsilon}$ does not hold, then $\mathsf{XOR}_{m, t, 0}$ does not hold. If the $\mathsf{XOR}_{m, t, \varepsilon}$ assumption does not hold, then there exists a polynomial time distinguisher $\mathcal{A}$ which can distinguish between $\mathcal{D}_0(n, m)$ and $\mathcal{D}_1(n, m, t, \varepsilon)$ with non-negligible advantage $p(n)$. Consider now the distinguisher $\mathcal{A}'$, which on an input $G \in \F_q^{n \times m}$, samples $i \leftarrow [1, n], v \leftarrow \text{Ber}(m, \varepsilon)$, creates a new matrix $G' \in \F_q^{n \times n}$ where $G' = G$, sets $G'_i \leftarrow G'_i \oplus v$, and then outputs $\mathcal{A}(G')$.\

    Notice first that $\mathcal{A}'$ is a polynomial time distinguisher since constructing $G'$ and running $\mathcal{A}(G')$ are efficient computations. We now show that $\mathcal{A}'$ can distinguish between $\mathcal{D}_0(m, n)$ and $\mathcal{D}_1(m, n, t, 0)$ with non-negligible advantage. Let us first consider the case when $G$ is sampled from $\mathcal{D}_0(m, n)$. In this case $G'$ is distributed as $\mathcal{D}_0(m, n)$. 
    
    Now let us consider the case when $G$ is sampled from $\mathcal{D}_1(m, n, t, 0)$. We now define three families of distributions:
    \begin{enumerate}
        \item
         For any $s \in \{ 0, 1\}^n$, let $\mathcal{D}^s_0(m, n)$ be distribution of $G$ when it is sampled uniformly from $\{ 0, 1\}^{n \times m}$ subject to $s^T G = 0$.
         \item 
         For any $s \in \{ 0, 1\}^n$, let $\mathcal{D}^s_1(m, n)$ denote the distribution of $G$ when it is sampled as follows: sample $v \leftarrow \text{Ber}(n, \varepsilon)$, sample $G$ uniformly from $\{ 0, 1\}^{n \times m}$ subject to $s^T G = v$. Also denote $\mathcal{D}_1(m, n) = \mathcal{D}_1(m, n, t, \varepsilon)$. 
         \item 
         Let $\mathcal{D}^s_2(m, n)$ denote the distribution of $G$ when sample it as follows: sample $i \leftarrow [1, n]$, $v \leftarrow \text{Ber}(n, \varepsilon)$, $G \leftarrow \mathcal{D}_0^s(m, n)$, set $G_i \leftarrow G_i \oplus v$. Let $\mathcal{D}_2(m, n)$ denote the distribution of $G$ as follows: Sample $s \leftarrow \mathcal{S}_{t, n}$ and $G \leftarrow \mathcal{D}^s_2(m, n)$.
    \end{enumerate}
    Notice that $\mathcal{D}_2^s(m, n) = (t/n) \mathcal{D}^s_1(m, n) + (1-t/n) \mathcal{D}^s_0(m, n)$ since there is a $t/n$ chance that $s_i = 1$, in which case $G$ is sampled from $\mathcal{D}^s_1(m, n)$. Therefore, the distribution of $\mathcal{D}_2(m, n)$ is
    \begin{align*}
        \mathcal{D}_2(m, n) &= \sum_{\substack{s \\ |s| = t}} \frac{\mathcal{D}_2^s(m, n)}{{n \choose t}}\\
        &= \sum_{\substack{s \\ |s| = t}} \left( \frac{t}{n} \frac{\mathcal{D}_1^s(m, n)}{{n \choose t}} + \left( 1-\frac{t}{n} \right) \frac{ \mathcal{D}_0^s(m, n)}{{n \choose t}} \right)\\
        &= \frac{t}{n} \sum_{\substack{s \\ |s| = t}} \frac{\mathcal{D}_1^s(m, n)}{{n \choose t}} + \left( 1-\frac{t}{n} \right) \sum_{\substack{s \\ |s| = t}} \frac{\mathcal{D}_0^s(m, n)}{{n \choose t}}\\
        &= \frac{t}{n} \mathcal{D}_1(m, n) + \left( 1-\frac{t}{n} \right) \sum_{\substack{s \\ |s| = t}} \mathcal{D}_0(m, n)
    \end{align*}

    In our reduction, if $G \leftarrow \mathcal{D}_0(m, n)$, then $G' \sim \mathcal{D}_0(m, n)$, and if $G \leftarrow \mathcal{D}_1(m, n, t, \varepsilon)$, then $G' \sim \mathcal{D}_2(m, n)$. The distinguishing advantage of $\mathcal{A}$ on these two distributions is
    \begin{align*}
        &\left| \underset{
        x \leftarrow \mathcal{D}_0(m, n)}{\operatorname{Pr}}[\mathcal{A}(x) = 1] - 
        \underset{
        x \leftarrow \mathcal{D}_2(m, n)}{\operatorname{Pr}}[\mathcal{A}(x) = 1] \right|\\
        =& \left| \underset{
        x \leftarrow \mathcal{D}_0(m, n)}{\operatorname{Pr}}[\mathcal{A}(x) = 1] - 
        \frac{t}{n} \underset{
        x \leftarrow \mathcal{D}_1(m, n)}{\operatorname{Pr}}[\mathcal{A}(x) = 1] - \left( 1 - \frac{t}{n}\right) \underset{
        x \leftarrow \mathcal{D}_0(m, n)}{\operatorname{Pr}}[\mathcal{A}(x) = 1] \right|\\
        =& \frac{t}{n} \left| \underset{
        x \leftarrow \mathcal{D}_0(m, n)}{\operatorname{Pr}}[\mathcal{A}(x) = 1] - 
        \underset{
        x \leftarrow \mathcal{D}_1(m, n)}{\operatorname{Pr}}[\mathcal{A}(x) = 1] \right|\\
        =& \frac{t}{n} p(n)
    \end{align*}
    
    Since $f(n)$ is non-negligible, $(t/n) p(n)$ is non-negligible. Therefore $\mathcal{A}$ distinguishes $\mathcal{D}_0(m, n)$ (the distribution of $G'$ when $G$ comes form $\mathcal{D}_0(m, n)$) and $\mathcal{D}_2(m, n)$ (the distribution of $G'$ when $G$ comes form $\mathcal{D}_1(m, n, t, \varepsilon)$) with non-negligible probability. So $\mathcal{A}'$ is a distinguisher falsifying the $\mathsf{XOR}_{m, t, \varepsilon}$ assumption.
    }
    {See full version.}
\end{proof}

Our second piece of evidence that $\mathsf{XOR}_{m, t, \varepsilon}$ is a weaker assumption than $\mathsf{XOR}_{m, t, 0}$ is that $\mathsf{XOR}_{m, t, \varepsilon}$ seems more robust to known attacks than $\mathsf{XOR}_{m, t, 0}$. The first version of \cite{christ2024pseudorandom} assumed $\mathsf{XOR}_{\Theta(n), t, 0}$. However, a subsequent version of \cite{sparse-k-sum} gave an attack showing $\mathsf{XOR}_{\Theta(n), O(\log n), 0}$ is not true. The attack (Thm 4.26 of \cite{sparse-k-sum}) consists of sampling random $m/2 \times m$ submatricies of the input matrix $G$ and then using Gaussian elimination to determine the submatrix contains a sparse subset of rows which xor to zero. The newest version of \cite{christ2024pseudorandom} circumvents this problem by setting $m = n^{1-\Omega(1)}$. 
We note that while $\mathsf{XOR}_{\Theta(n), O(\log n), 0}$ is susceptible to this type of attack, $\mathsf{XOR}_{\Theta(n), O(\log n), \varepsilon}$ is not for reasonable $\varepsilon$ (say $\varepsilon = 1/\sqrt{m}$). When attempting this attack against $\mathsf{XOR}_{\Theta(n), O(\log n), 0}$, we can use Gaussian elimination since we were looking for a zero vector in a $m/2$ dimensional subspace. When attempting this attack against $\mathsf{XOR}_{\Theta(n), O(\log n), \varepsilon}$, we must find a low weight vector in a $m/2$ dimensional subspace. This problem is the average case version of the problem finding a planted low weight codeword $v$ in a linear code, a problem which is generally believed to be intractable.

\subsection{The construction}
\begin{construction}[{Weak sparse xor construction, \textsf{weakXOR}[$m, t, \varepsilon, \eta$]}]
\label{cons: enc_LPN}
Let $m, t, \varepsilon, \eta$ be efficiently computable functions of the security parameter $n$
\begin{itemize}
    \item
    \textsf{KeyGen}$(1^n)$: Sample $(G, s)$ from $\mathcal{D}_1(n, m, t, \varepsilon)$. Output $(\textsf{sk} = s, \textsf{pk}=G)$.
    \item
    \textsf{Encode}$(1^n, G)$: Sample $u \leftarrow \text{Ber}(m, \eta)$, $e \leftarrow \text{Ber}(n, \eta)$. Output $Gu + e$.
    \item 
    \textsf{Decode}$(1^n, s, x)$: If $s^T x = 0$, output $1$. Otherwise, output $\bot$.
\end{itemize}
\end{construction}

\subsection{Robustness}
\begin{lemma}
    \label{lem: weakxor_robustness}
    Let $m = \poly(n), \eta = o(1), t = O(\log n), \varepsilon = O(\log(m)/(\eta m))$, and $p$ be any constant in $[0, 1/2)$. There exists a polynomial $p(n)$ such that for \cref{cons: enc_LPN}, for any $d \leq pn$, for a $1-\negl(n)$ fraction of keys $(\mathsf{pk}, \mathsf{sk}) \leftarrow \mathsf{KeyGen}(1^n)$,
        $$\underset{\mathcal{E}}{\operatorname{Pr}}[\mathsf{Decode}(1^n, \mathsf{sk}, \mathcal{E}(x)) = 1: x \leftarrow \mathsf{Encode}(1^n, \mathsf{pk}, 1)] \geq \frac{1}{2} + \frac{1}{p(n)}$$
        where $\mathcal{E}$ is the $d$-hypergeometric channel. the randomness is over the randomness of the encoding algorithm and the errors of $\mathcal{E}$.
\end{lemma}
\begin{proof}
    \fullornot{
    Let $p(n) = p'(n) \cdot p''(n) \cdot p'''(n)$ where $p'(n)$, $p''(n)$, and $p'''(n)$ are polynomials we will choose later. Consider the key sampling procedure and let $s^T G = v$. Notice that by \cref{lem: chernoff}, there is a $e^{-\Omega(\varepsilon m)}$ chance that $|v| \geq 1.5 \varepsilon m$. As long as $\varepsilon = \omega(\log(m)/m)$, this means there is a $1-\negl(n)$ chance (over the key sampling procedure) that $|v| \leq 1.5 \varepsilon m$.  Let $p'(n)$ be a polynomial so that $(1-2\eta)^{1.5 \varepsilon m} \geq 1/p'(n)$. We will assume for the remainder of the proof that $(1-2\eta)^{|v|} \geq 1/n^c$ since this happens for a $1-\negl(n)$ fraction of keys.
    
    Fix the keys. We must show that over the randomness encoding function and the channel, that a codeword from the PRC has a reasonable chance of being decoded to one. Formally, we wish to show that
    $$\underset{u, e, \mathcal{E}}{\operatorname{Pr}}[\textsf{Decode}(1^n, s, \mathcal{E}(Gu + e) = 1] = \frac{1}{2} + \frac{1}{p(n)}$$
    Since $\mathcal{E}$ is the $d$-hypergeometric channel, we can model the error from the channel as an error vector $e' \leftarrow \mathcal{S}_{d, n}$. We then need to prove
    $$\underset{u, e, e'}{\operatorname{Pr}}[\textsf{Decode}(1^n, s, Gu + e + e') = 1] = \frac{1}{2} + \frac{1}{p(n)}$$
    By the definition of our decoding function, the above probability is equal to
    $$\underset{u, e, e'}{\operatorname{Pr}}[s^T(Gu + e + e') = 0]
    = \underset{u, e, e'}{\operatorname{Pr}}[(s^T G)u + s^Te + s^T e') = 0]$$

    Since $(s^T G)u = v^T u$ is the xor of the coordinates of $u$ on which $v$ is one, $v^T u$ is the xor of $|v|$ i.i.d $\text{Ber}(\eta)$ random variables. Therefore, by \cref{lem: xoring_bias}, $(s^T G)u = v u$ has probability $1/2+(1-2\eta)^{|v|} = 1/2 + 1/p'(n)$ of being zero. Similarly, let $p''(n)$ be a polynomial (which we know exists by \cref{lem: xoring_bias}) so that $s^T e$ has at least a $1/2 + 1/p''(n)$ chance of being $0$. Let $p'''(n)$ be a sufficiently large constant so that $1/2 + (1/2) (pn/(n-t))^t \geq 1/2 + 1/p'''(n)$. Finally, by \cref{lem: xor_bias_no_replacement}, $s^T e'$ also has at least $1/2 + (1/2) (pn/(n-t))^t \geq 1/2 + 1/p'''(n)$ chance of being $0$. Since these events are all independent $s^T G u + s^T e + s^T e'$ has probability at least $1/2 + 1/p(n)$ of being $0$ by \cref{lem: xoring_bias}.
    }
    {See full version.}
\end{proof}

\subsection{A form of soundness}
\begin{lemma}
    \label{lem: weakxor_pseudosoundness}
    For any $m, t, \varepsilon, \eta$, $k = \poly(n)$, for \cref{cons: enc_LPN}, for all key pairs $(\mathsf{pk}, \mathsf{sk}) \leftarrow \mathsf{KeyGen}(1^n)$, we have
    $$\underset{x \leftarrow \{ 0, 1\}^n}{\operatorname{Pr}}[\mathsf{Decode}(1^n, \mathsf{sk}, x) = \bot] = \frac{1}{2}$$
\end{lemma}
\begin{proof}
    \fullornot{
    Recall that decoding in \cref{cons: enc_LPN} outputs $\bot$ if some fixed $O(\log n)$ bits (depending on $\mathsf{sk}$) of $x$ xor to $1$. For a random string $x$, this happens with probability exactly $1/2$.
    }
    {See full version.}
\end{proof}

\subsection{Pseudorandomness}
\begin{lemma}
    \label{lem: weakxor_pseudorandomness}
    For any efficiently computable $m = \poly(n), t, \varepsilon, \eta$, if $\textsf{XOR}_{m, t, \varepsilon}$ holds, and $\mathsf{LPN}[\eta]$ holds, then $\mathsf{weakXOR}[m, t, \varepsilon, \eta]$ is pseudorandom.
\end{lemma}
\begin{proof}
    \fullornot{
    Let $G$ be distributed uniformly over $\{ 0, 1\}^{n \times m}$, $G'$ be distributed according to $\mathcal{D}_1(m, n, t, \varepsilon)$, $q = \poly(n)$, $x_1, \dots, x_q$ be distributed as $\text{Ber}(m, \eta)$, and $e_1, \dots, e_q$ be distributed as $\text{Ber}(m, \eta)$.
    The output distribution of $\mathsf{weakXOR}[m, t, \varepsilon, \eta]$ is $G' x_1 + e_1, \dots, G' x_q + e_q$. By the $\textsf{XOR}_{m, t, \varepsilon}$ assumption, that distribution is indistinguishable from $G x_1 + e_1, \dots, G x_q + e_q$. The $\mathsf{LPN}[\eta]$ assumption implies $\mathsf{LPN'}[\eta]$, which implies that $G x_1 + e_1, \dots, G x_q + e_q$ is indistinguishable from random. By the hybrid lemma, $G' x_1 + e_1, \dots, G' x_s + e_s$ is also indistinguishable from random.
    }
    {See full version.}
\end{proof}
\begin{remark}
\label{remark: weakenedLPN}
Technically, we require something weaker than $\mathsf{LPN}[\eta]$ to hold for our proof of pseudorandomness. We need only that LPN is secure when the distribution of the secret comes from $\text{Ber}(m, \eta)$ and the error comes from a distribution $\text{Ber}(n, p)$ for any constant $p < 1/2$. However, since the LPN assumption is typically stated solely in terms of the error rate and $\mathsf{LPN}[\eta]$ is sufficient for this construction, we choose to state our results as being based on the (possibly stronger than necessary) $\mathsf{LPN}[\eta]$ assumption. 
\end{remark}

\subsection{Putting it all together}
\begin{proof}[Proof of \cref{thm: weakxor_bigtm}]
    \fullornot{
    Observe that the length of the codewords generated by \cref{cons: enc_LPN} is $\poly(n)$. Furthermore, \cref{lem: weakxor_robustness}, \cref{lem: weakxor_pseudosoundness}, and \cref{lem: weakxor_pseudorandomness} give us the remaining necessary preconditions to apply \cref{lem: amplify_public} which gives $(1-\negl(n), 1-\negl(n), \negl(n))$-public-key PRC.
    }
    {See full version.}
\end{proof}

\section{PRCs for space-bounded adversaries}
\label{sec: time-space}
We now present a zero-bit PRC scheme based on the time-space hardness of the learning parity with noise problem which is robust $\text{BSC}(p)$ for any constant $p<1/2$. The pseudorandomness of this construction is \emph{unconditional} and not based on cryptographic assumptions.

\begin{theorem}
    \label{thm: time-space-main-thm}
    Let $0 \leq \delta \leq 1/100$ be a constant and $p$ be a constant in $[0, 1/2)$. There exists a constant $c >0$ such that $\mathsf{SSR}[c \log(n), \varepsilon, k, k', \delta]$ is a zero-bit secret-key PRC which 
    \begin{itemize}
        \item 
        has output length $O(n)$
        \item
        is robust to BSC($p$)
        \item 
        has key size $O(n)$
        \item
        pseudorandom against probabilistic polynomial time, $O(n^{1.5-2 \delta}/\log^{0.01}(n))$ space adversaries.
    \end{itemize}
\end{theorem}

The celebrated work of \cite{RazLearningRequiresGoodMemory} showed that the learning parity without noise problem requires either a superpolynomial number of samples or $\Omega(n^2)$ memory. Follow-up work \cite{koltime-space} and \cite{garg2021memorysample} expanded this work to the cases where the secret is sparse and the case where the samples are noisy. We will begin by reviewing the relevant definitions and results.
\begin{definition}
    The learning sparse parities problem with density $\ell$ and error rate $\varepsilon$ is defined as follows: The secret vector $s$ is sampled uniformly at random from $\mathcal{S}_{\ell, n}$. An algorithm $\mathcal{A}$ is given samples $(a, a \cdot s + e)$ where $a \leftarrow \{ 0, 1\}^n, e \leftarrow \text{Ber}(1/2 - \varepsilon)$. We say $\mathcal{A}$ succeeds if it successfully outputs $s$.
\end{definition}

\fullversion{
\begin{definition}
    We say that a distribution of bits $X_1, \dots, X_n$ is next-bit unpredictable for a class of adversaries $\C$ if for all $\mathcal{A} \in \C$ and all $i \in [1, n]$, there exists a negligible function $\varepsilon(n)$ such that
    $$
    \underset{}{\operatorname{Pr}} \left[ \mathcal{A}(1^n, X_1, \dots, X_{i-1}) = X_i \right] \leq \frac{1}{2} + \varepsilon(n)
    $$
\end{definition}
}

\begin{lemma}
    \label{lem: TShardness}
    Let $q=\poly(n)$ and $\varepsilon = o(1)$. The distribution $a_1, a_1 \cdot s + e_1, \dots, a_q, a_q \cdot s + e_q$ where $s \leftarrow \mathcal{S}_{\Theta(\log n), n}$ and $a_i \leftarrow \{ 0, 1\}^n, e_i \leftarrow \text{Ber}(1/2-\varepsilon)$ for all $i \in [1, q]$ is next-bit unpredictable for PPT algorithms with $O(n \log^{0.99}(n)/\varepsilon)$ space.
\end{lemma}

\fullversion{See \cref{sec: time-space-rederivation} for a derivation of \cref{lem: TShardness}.}

\subsection{Construction}

\begin{construction}[{small space resilient construction, \textsf{SSR}[$ \ell, \varepsilon, k, \delta$]}]
\label{cons: enc_time-space}
Let $\ell, \varepsilon, k'$ be efficiently computable functions of the security parameter $n$ and $\delta$ be a constant.
\begin{itemize}
    \item
    \textsf{KeyGen}$(1^n)$: Sample $s_{1}, \dots, s_{k'} \leftarrow \mathcal{S}_{\ell, n}$ and output $\textsf{sk} = (s_{1}, \dots, s_{k'})$.
    \item
    \textsf{Encode}$(1^n, (s_{1}, \dots, s_{k'}), 1)$: Sample $a \leftarrow \{ 0, 1\}^n$, $e_{1}, \dots, e_{k'} \leftarrow \text{Ber}(1/2-\varepsilon)$, output 
    $$a || a \cdot s_{1} + e_1 || \dots || a \cdot s_{k'} + e_{k'} \ . $$
    
    \item 
    \textsf{Decode}$(1^n, (s_{1}, \dots, s_{k'}), x)$: Reinterpret $x \in \{ 0, 1\}^{n + k'}$ as $\Tilde{a} || \Tilde{b}_1 || \dots || \Tilde{b}_{k'}$ where $\Tilde{a} \in \{ 0, 1\}^n$ and $\Tilde{b}_{i} \in \{ 0, 1\}$ for all $i \in [1, k']$. If $\Tilde{a}$ is not $1/(2 n^{0.4})$ balanced, output $\bot$. Otherwise, let $w_i$ be one if and only if $\Tilde{a} \cdot s_{i} = \Tilde{b}_{i}$. If $\sum_{i=0}^{k'} w_i \geq k'/2 + n^\delta \sqrt{k'}$ output $1$ and otherwise output $\bot$.
\end{itemize}
\end{construction}

\subsection{Robustness}

Say that the decoder receives a string $x = \Tilde{a} || \Tilde{b}_1 || \dots || \Tilde{b}_{k'}$. Intuitively, for every $i$ such that $\Tilde{a} \cdot s_{i} = \Tilde{b}_{i}$, the decoder gains more confidence that $x$ is a codeword. However, on first inspection, it seems plausible one could flip a just a few of the first $n$ bits of a codeword (turn $a$ into $\Tilde{a}$) to ensure there would exist very few $i \in [k']$ such that $\Tilde{a} \cdot s_{i} = \Tilde{b}_{i}$. The existence of such an attack could potentially imply that the code of \cref{cons: enc_time-space} is not particularly robust to errors. We will show that such an attack does not affect robustness due to the sparsity of the $s_i$.
\fullornot{In order to do so, we first review a version of the Chernoff bound for weakly dependent random variables.}

\begin{definition}[\cite{Gavinsky2012ATB}]
    \label{def: read-k}
    A family $Y_1, \dots, Y_{k'}$ of random variables is read-$d$ if there exists a sequence $X_1, \dots, X_n$ of independent variables, and a sequence $S_1, \dots, S_{k'}$ of subsets of $[n]$ such that
    \begin{enumerate}
        \item Each $Y_i$ is a function of $(X_j: j \in S_i)$, and
        \item No element of $[n]$ appears in more than $d$ of the $S_i$'s.
    \end{enumerate}
\end{definition}

\begin{lemma}[\cite{Gavinsky2012ATB}]
    \label{lem: gavinsky}
    Let $Y_1, \dots, Y_{k'}$ be a family of read-$d$ indicator random variables with $\operatorname{Pr}[Y_i = 1] = p_i$ and let $p$ be the average of $p_1, \dots, p_{k'}$. Then for any $\varepsilon > 0$, the probabilities
    $$\operatorname{Pr}[Y_1 + \dots + Y_{k'} \geq (p+\varepsilon)k'] \qquad \text{and} \qquad
    \operatorname{Pr}[Y_1 + \dots + Y_{k'} \leq (p-\varepsilon)k']$$
    are both at most $e^{-2 \varepsilon^2 {k'}/d}$
\end{lemma}

\begin{lemma}
    \label{lemma: bounded_dependence}
    Let $\ell \leq O(\log n)$, $d = \omega(\log n)$, and $k' \leq n$. Consider $S = \{ S_1, \dots, S_{k'} \}$ where each $S_i$ is drawn uniformly at random from ${[n] \choose \ell}$. Some element $t \in [n]$ occurs in $d$ elements of $S$ with probability $\negl(n)$.
\end{lemma}
\begin{proof}
    \fullornot{
    Let $T_t$ denote the event where $i$  occurs in at least $d$ elements of $S$. Since $t$ occurs in each element of $S$ independently with probability $\ell/n$, the probability that it appears in at least $d$ elements of $S$ is the probability that a random variable distributed as $\text{Bin}(k', \ell/n)$ is at least $d$. Since $k' \leq n$ and $d = \omega(\log n)$, this is at most the probability that a random variable distributed as $\text{Bin}(n, c \log(n)/n)$ is at least $\omega(\log n)$. By \cref{lem: chernoff}, this probability is at most $\negl(n)$.  Union bounding over all $t \in [n]$, we see that probability that there exits some element $t \in [n]$ occurring in more than $d$ sets is at most $n \cdot \negl(n) = \negl(n)$.
    }
    {See full version.}
\end{proof}

\begin{lemma}
    \label{lem: time_space_robust}
   Let $\varepsilon$ be some function of $n$, $p$ be a constant in $[0, 1/2)$, $\delta > 0$, and $k' = (2 n^{2\delta}/\varepsilon)^2$. There exists a constant $c >0$ such that for $\ell = c \log(n)$, $\mathsf{SSR}[\ell, \varepsilon, k', \delta]$ is robust to BSC($p$) with probability $1-\negl(n)$.
\end{lemma}
\begin{proof}
    \fullornot{
    The probability that a codeword is decoded to $1$ correctly is equal to the probability that the following experiment succeeds.
    We sample $s_{1}, \dots, s_{k'} \leftarrow \mathcal{S}_{\ell, n}$, $e_1, \dots, e_{k'} \sim \text{Ber}(1/2-\varepsilon)$, $a \leftarrow \{ 0, 1\}^n$, and $e'$ from $\text{Ber}(n+k', p)$. Let $\Tilde{a} = a \oplus e'_{[1,n]}$, $\Tilde{b}_{i} = a \cdot s_{i} + e_i + e'_{n+i}$ for every $i \in [1, k']$, and $w_i = 1$ if and only if $\Tilde{a} \cdot s_{i}  = \Tilde{b}_{i}$.  The experiment succeeds if $\sum_{i=1}^{k'} w_{i} \geq k'/2 + n^\delta \sqrt{k'}$.
    
    By \cref{lemma: bounded_dependence}, we can fix $S = \{ S_1, \dots, S_{k'} \}$ and assume that no element $t \in [n]$ occurs in more than $n^\delta/2 = \omega(\log n)$ elements of $S$ since this is true with $1-\negl(n)$ probability.
    
    We see that each $w_{i}$ is a function of $\{ a_r : r \in S_j \} \cup \{ e'_r : r \in S_j\} \cup \{ e_i, e'_{n+i} \}$. Since we assumed no element $t \in [n]$ occurs in more than $n^\delta/2$ of the sets $S_j$, we see that no $w_{i}, w_{i'}$ where $i \neq i'$ share more than $n^\delta$ random variables on which they are dependent. Therefore, $w_{i}$ are read-$n^\delta$ random variables. 
    We will now compute the expectation of $w_{i}$:
    \begin{align*}
        \mathbb{E}[w_i] &= {\operatorname{Pr}}[\Tilde{a} \cdot s_{i} = \Tilde{b}_{i}]\\
        &= {\operatorname{Pr}}[(a+e'_{[1,n]}) \cdot s_i = (a \cdot s_i) + e_i + e'_{n+i}]\\
        &= {\operatorname{Pr}}[e'_{[1,n]} \cdot s_i = e_i + e'_{n+i}] \ .
    \end{align*}
    Recall $|s_i| = O(\log n)$, and by symmetry, we can assume without loss of generality that $s_i$ is a series of ones followed by a series of zeros. So the above probability expression is equal to
    \begin{align*}
        &= {\operatorname{Pr}}[e'_{[1, c \log(n)]} + e_i + e'_{n+i} = 0]\\
        &= \frac{1}{2} \left( 1 + (1-2p)^{c \log(n)+1} \left(1-2\left(\frac{1}{2} - \varepsilon \right) \right) \right)\\
        &= \frac{1}{2} \left( 1 + 2 \varepsilon (1-2p)^{c \log(n)+1} \right)
    \end{align*}
    where the second equality follows from \cref{lem: xoring_bias}.
    We can set $c$ to be a sufficiently small constant such that $(1-2p)^{c \log(n)+1} \geq 1/n^{\delta}$ so that the above is at least $\frac{1}{2} + \varepsilon/n^{\delta}$.
    
    Now that we know the expected value for each $w_{i}$, we can use the Chernoff bound for variables with bounded dependence. For sufficiently large $n$, there are $k'$ such read-$n^\delta$ variables and each has probability at least $1/2 + \varepsilon/n^{\delta}$ of being $1$. The probability that the decoding algorithm decodes to $\bot$ is
    \begin{align*}
        {\operatorname{Pr}}\left[w_{1} + \dots + w_{k'} \leq k'/2 + n^\delta \sqrt{k'} \right]
        &= {\operatorname{Pr}}\left[\sum_{i=1}^{k'} w_{i} \leq (1/2 + \varepsilon/n^{\delta} - \varepsilon/n^{\delta} + n^\delta/\sqrt{k'})k' \right]\\
        &\leq e^{-2 (-\varepsilon/n^{\delta} + n^\delta/\sqrt{k'})^2 k'/n^\delta}\\
        &= e^{-\Omega(\varepsilon/n^\delta)^2 k'/n^\delta}\\
        &= e^{-\Omega(\varepsilon^2/n^{2\delta}) k'/n^\delta}\\
        &= e^{-\Omega(\varepsilon^2/n^{2\delta}) \cdot \Omega(n^{4 \delta}/\varepsilon^2)/n^\delta}\\
        &= e^{-\Omega(n^{2 \delta})/n^\delta}\\
        &= \negl(n)
    \end{align*}
    where the second inequality follows from \cref{lem: gavinsky} and the third follows by assumption on the value of $k'$.
    Therefore, there is a $\negl(n)$ chance that codeword is decoded to $\bot$.
    }
    {See full version.}
\end{proof}

\subsection{Soundness}
On first inspection, it may seem strange that we output $\bot$ when trying to decode strings where $\Tilde{a}$ is not balanced. This is to ensure soundness. To see why this exit condition is necessary, consider what happens when the codeword is the string of all zeros. \cref{cons: enc_time-space} would certainly decode this codeword to $0$ regardless of what $\mathsf{sk}$ is. The requirement that $\Tilde{a}$ eliminates the possibility of such edge cases.
We will now show the soundness of our zero bit encryption scheme by showing that any fixed $x \in \{ 0, 1\}^{n+k'}$ decodes to $\bot$ with high probability.

\begin{lemma}
    \label{lem: decoding_partial_key}
    Let $a \in \{ 0, 1\}^n$ be a $1/n^{0.4}$-biased string, $b \in \{ 0, 1\}$, $c$ be an arbitrary constant and $s$ be drawn uniformly at random from $\mathcal{S}_{c \log(n), n}$.
    $$\underset{s}{\operatorname{Pr}} \left[ a \cdot s = b \right] \leq 1/2 + \negl(n)$$
\end{lemma}
\begin{proof}
    \fullornot{
    Let $r = |\{ i : a_i = 1\}|$ and notice $n/2-n^{0.6} \leq r \leq n/2+n^{0.6}$ since $a$ is $\sigma$-balanced. The probability that $a \cdot s = 0$ is the probability $X \sim \text{Hyp}(n, r, c \log(n))$ is even. Since $1/2 - O(1/n^{0.1}) \leq (r-c \log(n))/n \leq 1/2 + O(1/n^{0.1})$ and $1/2 - O(1/n^{0.1}) \leq r/(n-c \log n) \leq 1/2 + O(1/n^{0.1})$, by \cref{cor: xor_bias_no_replacement}, the probability that $X \sim \text{Hyp}(n, r, c \log(n))$ is even is at most $1/2+O(1/n^{0.1})^{c \log(n)} = 1/2 + \negl(n)$. A similar argument shows that the probability that $a \cdot s = 1$ is at most $1/2 + \negl(n)$.
    }
    {See full version.}
\end{proof}

\begin{lemma}
    \label{lem: time-space-sound}
    Let $0 \leq \delta \leq 1/100$ be constant, $\epsilon$ be any function of $n$, $k' \geq n^\delta$ be $\poly(n)$, and $\ell = O(\log n)$. For any fixed $x \in \{ 0, 1\}^{n+k'}$, in the $\mathsf{SRR}[\ell, \varepsilon, k', \delta]$ scheme,
    $$\underset{\mathsf{sk}}{\operatorname{Pr}} \left[ \mathsf{Decode}(\mathsf{sk}, x) = \bot \right] \geq 1-\negl(n) \ . $$
\end{lemma}
\begin{proof}
    \fullornot{
    Let us reanalyze $x$ as $a || b_1 || \dots || b_{k'}$ where $a \in \{ 0, 1\}^n$ and $b_i \in \{ 0, 1\}$ for $i \in [1, k']$. Consider the set $G = \{ j : a \cdot s_j = b_j \}$. Recall that for $x$ to not decode to $\bot$, we need $|G| \geq k'/2 + n^\delta \sqrt{k'}$. By \cref{lem: decoding_partial_key}, each $i$ is in $G$ independently with probability $1/2 + \negl(n)$. So the mean value of $|G|$ is $k'/2 + \negl(n) k'$. Therefore
    \begin{align*}
        \underset{\mathsf{sk}}{\operatorname{Pr}} \left[ \mathsf{Decode}(\mathsf{sk}, x) \neq \bot \right]
        & =  \underset{\mathsf{sk}}{\operatorname{Pr}} \left[ |G| \geq \frac{k'}{2} + n^\delta \sqrt{k'}  \right]\\
        & \leq \underset{\mathsf{sk}}{\operatorname{Pr}} \left[ |G| \geq  \left( 1+\frac{n^{\delta/100}}{\sqrt{k'}} \right) \left( \frac{k'}{2}+\negl(n) \right) \right]\\
        &=\negl(n)
    \end{align*}
    where the second to last inequality is true for sufficiently large $n$ and the last inequality follows from \cref{lem: chernoff}.
    }
    {See full version.}
\end{proof}

\subsection{Pseudorandomness}
We will show pseudorandomness of \cref{cons: enc_time-space} by first showing that any polynomial number of codewords is next-bit unpredictable for a polynomial time, space-bounded adversary. \cref{lem: TShardness} shows that sparse parity learning examples $a || a \cdot s + e$ are next bit unpredictable. In this case, $a$ is random and one pseudorandom bit is output per freshly sampled $a$. However, in \cref{cons: enc_time-space}, the samples are of the form $a || a \cdot s_1 + e_1 || \dots || a \cdot s_{k'} + e_{k'}$. In this case, $a$ is random and multiple pseudorandom bits are output per freshly sampled $a$. Fortunately, next-bit unpredictability of samples of the form $a || a \cdot s + e$ implies next-bit unpredictability of samples of the form $a || a \cdot s_1 + e_1 || \dots || a \dots s_{k'} + e_{k'}$.

\begin{lemma}
    \label{lem: unpredictability_of_Enc}
    Let $0 \leq \delta \leq 1/100$ be a constant, $\varepsilon = 1/\poly(n)$, $\log(\varepsilon) \in \mathbb{Z}, q = \poly(n), k' = \poly(n)$, and $\ell = \Theta(\log n)$. Let $\mathsf{Enc}$ be the encoding function of $\mathsf{SSR}[\ell, \varepsilon, k', \delta]$. Consider the distribution induced by $\mathsf{sk} \leftarrow \mathsf{KeyGen}(1^n)$ and $X_1, \dots, X_q \leftarrow \mathsf{Enc}(1^n, \mathsf{sk}, 1)$ where $X_i \in \{ 0, 1\}^{n+k'}$ for all $i \in [1, q]$. No PPT, $O(n \log^{0.99}(n)/\varepsilon)$ space adversary acts as a next bit predictor for $X_1, \dots, X_q$.
\end{lemma}
\begin{proof}
    \fullornot{
    Consider the following distribution on $q(n+1)$ bits: $Y = (a_1, a_1 \cdot s + e_1, \dots, a_q, a_q \cdot s + e_q)$ where $s \leftarrow \mathcal{S}_{\ell, n}$ and $a_i \leftarrow \{ 0, 1\}^n, e_i \leftarrow \text{Ber}(1/2-\varepsilon)$. Say for the sake of contradiction that there exists an algorithm a polynomial time, $O(n \log^{0.99}(n)/\varepsilon)$ space adversary $\mathcal{A}$ and a series of indices $\{i_n\}_{n \in \mathbb{N}}$ such that $\mathcal{A}$ could predict bit $i_n$ of $(X_1, \dots, X_q) \in \{ 0, 1\}^{q(n+k')}$ with non-negligible probability. We will use $\mathcal{A}$ to construct an algorithm $\mathcal{A}'$ which acts as a next bit predictor for $Y$ by predicting bit $\{i'_n\}_{n \in \mathbb{N}}$ of $Y$ with non-negligable probability.

    Bit $i_n$ cannot be a truly random bit belonging to a newly sampled $a$ since then it would not be predictable. Therefore, bit $i_n$ corresponds to $(a_g \cdot s_j + e_{g, j})$ for some $g \in [1, q], j\in [1, k']$. In words, $i_n$ is the $j$th parity bit from codeword $g$. We now construct $\mathcal{A}'$ that will predict bit $i'_n = j(n+1)$ of $Y$. $\mathcal{A}'$ begins by sampling each $s_{1}, \dots, s_{k'}$ \emph{excluding} $s_{j}$ from $\mathcal{S}_{\Theta(\log n), n}$. $\mathcal{A}'$ then simulates $\mathcal{A}$. 
    
    Recall $Y = (a_1, a_1 \cdot s + e_1, \dots, a_q, a_q \cdot s + e_q)$ and $\mathcal{A}'$ wishes to predict $j(n+1)$ of $Y$. $\mathcal{A}'$ samples $e_{u, v} \leftarrow \text{Ber}(1/2-\varepsilon)$ for all $u \in [q], v \in [1, k']$. Let $t_i \in \{ 0, 1\}^{n+k'}$ be $(a_i, a_i \cdot s_1 + e_{i, 1}, \dots, a_i \cdot s_{j-1} + e_{i, j-1}, a_i \cdot s + e_i, a_i \cdot s_{j+1} + e_{i, j+1}, \cdot s_{k'} + e_{i, k'})$. $\mathcal{A}'$ feeds the first $i_n-1$ bits of $(t_1, \dots, t_q)$ to $\mathcal{A}$, which then outputs it prediction $b$. $\mathcal{A}'$ outputs $b$.

    We now confirm that $\mathcal{A}'$ is a polynomial time, $O(n \log^{0.99}(n)/\varepsilon)$ space algorithm. The sampling of any $e \leftarrow \text{Ber}(1/2 - \varepsilon)$ requires some care. If $\varepsilon = 1/2^x$ for some $x = n^{O(1)}$, we can use $x$ bits to sample an integer $y$ uniformly at random from $[1, 2^x]$. If $y \in [1, 2^{x-1}+1]$, 
    we set $e = 0$ and otherwise set $e = 1$. This sampling procedure results in $e \sim \text{Ber}(1/2 - \varepsilon)$. The rest of the computation still clearly proceeds in $\poly(n)$ time. Furthermore, $\mathcal{A}'$ only needs $k' \cdot \log({n \choose \log(n)}) = k' \cdot O(\log^2(n)) = O(n)$ auxiliary space to store and compute with $(s_1, \dots, s_{k'})$. Therefore, $\mathcal{A}'$ is a $\poly(n)$ time, $O(n \log^{0.99}(n)/\varepsilon)$ space algorithm.

    Since $(t_1, \dots, t_q)$ has the exact same distribution as $(X_1, \dots, X_q)$, it should be clear that $\mathcal{A}$ predicts bit $i_n$ of $(t_1, \dots, t_q)$ non-negligible probability. Since by construction, that bit corresponds exactly to bit $j(n+1)$ of $Y$, we see that $\mathcal{A}'$ predicts bit $j(n+1)$ of $Y$ with non-negligible probability.

    Therefore, $\mathcal{A'}$ is a $\poly(n)$ time, $O(n \log^{0.99}(n)/\varepsilon)$ space algorithm which predicts bit $i'_n$ of the output of $Y$. This contradicts \cref{lem: TShardness}.
    }
    {See full version.}
\end{proof}

\begin{lemma}
    \label{lem: time-space-pseudorandom}
    Let $0 \leq \delta \leq 1/100$ be a constant, $\varepsilon = 1/\poly(n)$, $\log(\varepsilon) \in \mathbb{Z}, q = \poly(n), k' = \poly(n)$, and $\ell = \Theta(\log n)$. The scheme $\mathsf{SSR}[ \ell, \varepsilon, k', \delta]$ is pseudorandom against $O(n \log^{0.99}(n)/\varepsilon)$ space, $\poly(n)$ time adversaries.
\end{lemma}
\begin{proof}
    \fullornot{
    Follows from \cref{lem: unpredictability_of_Enc} and observing that the standard hybrid argument showing that next bit unpredictability implies pseudorandomness \cite{GoldreichFoundationsOfCrypto} applies even for space-bounded adversaries.
    }
    {See full version.}
\end{proof}

\subsection{Putting it all together}
\label{subsec: space_conclusion}
\begin{theorem}
    \label{thm: time-space-gen-thm}
    Let $0 \leq \delta \leq 1/100$ be a constant, $\varepsilon = 1/\poly(n), k' = (2 n^{2 \delta}/\varepsilon)^2$, and $p$ be a constant in $[0, 1/2)$. There exists a constant $c >0$ such that $\mathsf{SSR}[c \log(n), \varepsilon, k, k', \delta]$ is a zero-bit secret-key PRC which 
    \begin{itemize}
        \item 
        has output length $n+k'$
        \item
        is robust to BSC($p$)
        \item 
        has key size $O(k' \cdot \log^2(n))$
        \item
        pseudorandom against probabilistic polynomial time, $O(n \log^{0.99}(n)/\varepsilon)$ space adversaries.
    \end{itemize}
\end{theorem}
\begin{proof}
    \fullornot{
    Immediate by combining \cref{lem: time_space_robust}, \cref{lem: time-space-sound}, and \cref{lem: time-space-pseudorandom}.
    }
    {See full version.}
\end{proof}

\cref{thm: time-space-main-thm} shows that \cref{cons: enc_time-space} can have quite small key sizes at the expense of being pseudorandom against adversaries with smaller space. \cref{thm: time-space-main-thm} now follows by instantiating the parameter regime we believe to be the most useful.
\begin{proof}[Proof of \cref{thm: time-space-main-thm}]
    \fullornot{
    Follows by setting $\varepsilon = \log(n) n^{-1/2+2 \delta}$ in \cref{thm: time-space-gen-thm}.
    }
    {See full version.}
\end{proof}

Therefore, we have shown zero bit PRCs with $O(n)$ length which are \emph{unconditionally} pseudorandom against $\poly(n)$ time, $O(n^{1.5-\delta})$ space (for any constant $\delta>0$) adversaries. It is natural to ask if this leads to multi-bit PRCs. The construction of multi-bit PRCs with rate $1/n$ (construction 3 of \cite{christ2024pseudorandom}) also works in the space-bounded setting but has codeword length $O(k n)$ when encoding $k$ bits. This would let us build $k$-bit PRCs with codeword length $O(kn)$ which are pseudorandom against PPT, $O(n^{1.5-\delta})$ space adverseries. However, that construction has the undesirable property that it narrows the gap between the space of the adversary and the space of the encoding algorithm, thereby making the scheme less secure. 
It would be interesting to build constant rate PRCs which are unconditionally pseudorandom against PPT, space-bounded adversaries.

\section{Perspectives}
\label{sec: perspectives}
Here we review some of the design decisions we have made in our constructions.

In \cref{sec: time-space}, we prove robustness to the binary symmetric channel rather than $p$-bounded channels (we assume $p$ is a constant in $[0, 1/2)$). One may ask whether it is possible to prove robustness to  all $p$-bounded channels rather than just the binary symmetric channel. To show robustness to $p$-bounded channels, one could choose to apply a similar type of reduction as given in \cref{lem: amplify_public} by including in the secret key a shift $z$ and a permutation $\pi$. This reduces showing robustness against $p$-bounded channels to showing robustness against $d$-hypergeometric channels for all $d \leq pn$, which is very similar to the binary symmetric channel. Since the robustness probability only goes up as $p$ goes down in \fullornot{\cref{lem: time_space_robust}}{\cref{cons: enc_time-space}}, there exists a function $u(n) = \negl(n)$ such that for all $d \in [1, pn]$, \cref{cons: enc_time-space} is robust to $\text{BSC}(d/n)$ with probability $1-u(n)$. This implies \cref{cons: enc_time-space} is robust to any $d$-hypergeometric channel for $d \leq pn$ with probability $O(n) u(n) = \negl(n)$. However, such a reduction incurs an additive $O(n \log n)$ factor in the key size since $\pi$ is $O(n \log n)$ bits. In the space-bounded setting, having small keys is particularly important, so we have chosen to focus on the standard setting of the binary symmetric channel, which allows for remarkably small key sizes. However, it should not be hard to formalize the argument for $p$-bounded channels. Of course, one could use a pseudorandom function to generate $\pi$ and avoid the additive $O(n \log n)$ factor in the key size. We chose not to do this to keep our construction unconditional. Combining the results of \cref{sec: time-space} with other cryptographical objects (such as PRFs) remains an interesting open question.

This work focuses on the theoretical aspects of PRCs but one can also ask if \cref{sec: hyperloop} and \cref{sec: weakxor} are practical for watermarking LLM text. Unfortunately, this seems unlikely.
The problem is that if we set the security parameter $n = 128$ (a reasonable security parameter), the application of the \cref{lem: amplify_public}, which allows us to construct a PRC from a scheme where there is only a small advantage in distinguishing codewords from random words, requires us to concatenate many codewords together, which may result in a code with a length of $\poly(n)$ for some very large polynomial. This is too long to be practical. Fundamentally, \cref{lem: amplify_public} allows us to amplify robustness by concatenating $t$ codewords of length $n$ to form a string $x$ of length $tn$. Every $n$ bit block of $y = \mathcal{E}(x)$ that is decoded to $1$ rather than $\bot$ gives us more certainty that $y$ is a corrupted codeword. 

There are, however, other ways to amplify our confidence. For example, each codeword of length $n$ can contain multiple checks. In \cref{cons: enc_LPN}, (for simplicity, consider the $\epsilon=0$ regime) we sample $G$ uniformly at random subject to $s^T G = 0^m$ and then check if $y$ is a corrupted codeword by checking if $s^T y = 0$. This gives us low confidence that $y$ is a corrupted codeword, so we apply \cref{lem: amplify_public}. However, imagine we had $s_1, \dots, s_\tau$ and sampled $G$ uniformly at random subject to the constraints that $s_i^T G = 0^m$ for all $i \in [1, \tau]$. Then to check if $y$ is a corrupted codeword, we check how many $i \in [1, \tau]$ there were such that $s_i^T y = 0$, and the more there were, the more confidence that we could have that $y$ is a corrupted codeword. This is the approach advocated by \cite{christ2024pseudorandom}. 

Similarly, in \cref{sec: hyperloop}, we implant $\poly(n)$ hyperloops but one use one for decoding (by checking if $\bigoplus_{j \in S_1} y_{j} = 0$) and then amplify our success probability using \cref{lem: amplify_public}. From a theoretical perspective, the polynomial size blowup in the length of the code incurred by \cref{lem: amplify_public} does not matter. However, from a practical perspective, the correct approach would check how many $i \in [1, t]$ there are such that $\bigoplus_{j \in S_i} y_{j} = 0$. In both cases, adding more structure in the encoding/decoding stages means that the decoder knows with greater certainty if a word is a codeword, without incurring a large blowup in codeword length.

\bibliographystyle{alpha}
\bibliography{main}

\appendix

\fullversion{
\label{sec: time-space-rederivation}
\section{Time-space hardness of $\Theta(\log n)$-sparse LPN}
For our results in \cref{sec: time-space}, we need to confirm that any algorithm which solves the learning parity with noise where the secret has weight $\Theta(\log n)$ with non-negligible probability requires a significant amount of memory. We therefore begin by examining which parameter regimes we may expect negligible success probability (\cite{RazLearningRequiresGoodMemory} \cite{koltime-space} \cite{garg2021memorysample} are generally not concerned with the distinction between negligible and non-negligible success probability).
\begin{definition}[\cite{koltime-space}, this version appears in \cite{GargExtractorBasedTS}]
    A set $T \subseteq \{ 0, 1\}^n$ is $(\varepsilon, \delta)$-biased if there are at most $\delta \cdot 2^n$ elements $a \in \{ 0, 1\}^n$ with $|\mathbb{E}_{x \in_R T}[(-1)^{a \cdot x}]| > \varepsilon$, (where $a \cdot x$ denotes the inner product of $a$ and $x$, modulo 2).
\end{definition}

\begin{definition}[\cite{garg2021memorysample}]
    Let $X, A$ be two finite sets. A matrix $M: A \times X \rightarrow \{ -1, 1\}$ is a $(k,\ell)$-$L_2$-extractor with error $2^{-r}$ if for every nonnegative $f: X \rightarrow \mathbb{R}$ with $\| f\|_2/\| f\|_1 \leq 2^{\ell}$, there are at most $2^{-k} \cdot |A|$ rows $a \in A$ with 
    $$\frac{|\langle M_a, f \rangle|}{\| f\|_1} \geq 2^{-r}$$.
    Where $\| f \|_p = (\mathbb{E}_{x \leftarrow X}[|f(x)|^p])^{1/p}$, $M_a: X \rightarrow \mathbb{R}$ is the function corresponding to the $a$th row of $M$, and $\langle M_a, f \rangle = \mathbb{E}_{x \leftarrow X}[f(x) \cdot g(x)]$.
\end{definition}

\begin{lemma}[\cite{koltime-space}, this version appears in \cite{GargExtractorBasedTS}]
    \label{lem: ts-lem1}
    There exists a sufficiently small constant $c$ such that the following holds. Let $\mathcal{S}_{\ell, n} = \{ x \in \{ 0, 1\}^n : |x| = \ell \}$. If $\ell \leq n^{0.9}$, then $\mathcal{S}_{\ell, n}$ is an $(\varepsilon, \delta)$-biased set for $\varepsilon = \ell^{-c \ell}, \delta = 2^{-cn/\ell^{0.01}}$.
\end{lemma}

\begin{lemma}[\cite{GargExtractorBasedTS}]
    \label{lem: ts-lem2}
    Let $T \subseteq \{ 0, 1\}^n$ be an $(\varepsilon, \delta)$-biased set, with $\varepsilon \geq \delta$. Then the matrix $M: T \times \{ 0, 1\}^n \rightarrow \{ -1, 1\}$, defined by $M(a, x) = (-1)^{a \cdot x}$ is a $(k, \ell)$-$L_2$-extractor with error $2^{-r}$, for $\ell = \Omega(\log(1/\delta))$, and $k = r = \Omega(\log(1/\varepsilon))$.
\end{lemma}

\begin{lemma}[Theorem 5 of \cite{garg2021memorysample}]
    \label{lem: ts-lem3}
    Let $1/100 \leq c \leq \ln(2)/3$. Fix $\gamma$ to be such that $(3c)/\ln(2) \leq \gamma^2 \leq 1$. Let $X, A$ be two finite sets. Let $n = \log_2 |X|$. Let $M: A \times X \rightarrow \{ -1, 1\}$ be a matrix which is a $(k', \ell')$-$L_2$-extractor with error $2^{-r}$, for sufficiently large $k', \ell'$, and $r'$, where $\ell' \leq n$. Let 
    $$r := \min \left\{  \frac{r'}{2}, \frac{(1-\gamma) k'}{2}, \frac{(1-\gamma)\ell'}{2}-1 \right\} \ .$$
    Let $B$ be a branching program, of length at most $2^r$ and width at most $2^{c \cdot k' \cdot \ell'/\varepsilon}$ for the learning problem that corresponds to the matrix $M$ with error probability $\varepsilon$. Then the success probability of $B$ is at most $O(2^{-r})$.
\end{lemma}

\begin{lemma}
    No probabilistic polynomial time, $O(n \log^{0.99}(n)/\varepsilon)$ space algorithm solves the learning sparse parties problem with density $\Theta(\log n)$ and error rate $1/2-\varepsilon$ with non-negligible probability.
\end{lemma}
\begin{proof}
    We begin by observing that any probabilistic polynomial time, $O(n \log^{0.99}(n)/\varepsilon)$ space algorithm implies a deterministic polynomial time, $O(n \log^{0.99}(n)/\varepsilon)$ space algorithm since a deterministic algorithm can simply use the first bit of a freshly sampled LPN sample any time it needs a random bit. Therefore, we only need to show that no deterministic polynomial time, $O(n \log^{0.99}(n)/\varepsilon)$ space algorithm solves the learning sparse parties with density $\Theta(\log n)$ and error rate $1/2-\varepsilon$ with non-negligible probability. 
    
    Let $\ell = \Theta(\log n)$ be the weight of our LPN secret. \cref{lem: ts-lem1} tells us that there exists a constant $c$ such that $\mathcal{S}_{\ell, n}$ is an $(\varepsilon, \delta)$-biased set where $\varepsilon = \ell^{-c \ell}, \delta = 2^{-cn/\ell^{0.01}}$. Since $\varepsilon \geq \delta$ (for sufficiently large $n$), by \cref{lem: ts-lem2}, we see that the matrix $M$ defined by $\mathcal{S}_{\ell, n}$ is a $(k', \ell')$-$L_2$-extractor with error $2^{-r'}$ for $\ell' = \Omega(n/\ell^{0.01}), k' = r' = \Omega(\ell \log(\ell))$. Notice that $r', k', \ell'$ are all $\Omega(\ell \log(\ell))$. Therefore, by \cref{lem: ts-lem3}, we see that branching program which uses $o(k' \ell'/ \varepsilon) = o(\ell \log(\ell) n/(\ell^{0.01} \varepsilon)) = O(n \ell^{0.99} \log(\ell)/\varepsilon)$ memory and $O(2^{\ell \log(\ell)})$ time has as a $O(2^{-\ell \log(\ell)}) = \negl(n)$ of solving the learning parity with noise problem.
\end{proof}

Finally, as pointed out in \cite{koltime-space}, this gives us cryptography against a bounded space adversary. Since the inner product is a strong extractor, if we select our secret $s$ from $\mathcal{S}_{\ell, n}$ and output $(a_1, a_1 \cdot s + e_1), \dots, (a_{t}, a_{t} \cdot s + e_t), a_{t+1}$ (where $t = \poly(n)$, each $a_i$ is sampled uniformly from $\{ 0, 1\}^n$, and each $e_i$ is sampled from $\text{Ber}(1/2 - \varepsilon)$), no polynomial time, $o(n \ell^{0.99} \log(\ell)/\varepsilon)$ space algorithm can predict $a_{t+1} \cdot s$ with noticeable probability.

\section{Hypergeometric distribution lemma}
\label{sec: appendixB}
We restate and prove \cref{lem: xor_bias_no_replacement} and \cref{cor: xor_bias_no_replacement}.

\xorBiasHyp*
\begin{proof}
    Note that $X$ is the number of special elements chosen if have $n$ elements, of which $m$ are special, and we choose $t$. Consider selecting the elements one by one. Let $X_i$ be the indicator random variable denoting if the $i$th element chosen is special, let $a_i$ denote the probability that we have an even number of special elements after $i$ items are chosen. Let $p_i$ denote $\underset{X_1, \dots, X_n}{\operatorname{Pr}}[X_1 \oplus \dots \oplus X_{i+1} = 1 | X_1 \oplus \dots \oplus X_{i} = 0]$. We will show by induction that $a_i = 1/2 + (1 + \prod_{j=1}^i (1-2p_j))$. When $i = 0$, this holds trivially. We now show the inductive case.
    \begin{align*}
        a_{i+1} &= \underset{X_1, \dots, X_n}{\operatorname{Pr}}[X_1 \oplus \dots \oplus X_{i+1} = 0]\\
        &= \underset{X_1, \dots, X_n}{\operatorname{Pr}}[X_1 \oplus \dots \oplus X_{i+1} = 0 \cap X_1 \oplus \dots \oplus X_{i} = 0]\\
        &+ \underset{X_1, \dots, X_n}{\operatorname{Pr}}[X_1 \oplus \dots \oplus X_{i+1} = 0 \cap X_1 \oplus \dots \oplus X_{i} = 1]\\
        &= \underset{X_1, \dots, X_n}{\operatorname{Pr}}[X_1 \oplus \dots \oplus X_{i+1} = 0 | X_1 \oplus \dots \oplus X_{i} = 0] \cdot a_i\\
        &+ \underset{X_1, \dots, X_n}{\operatorname{Pr}}[X_1 \oplus \dots \oplus X_{i+1} = 0 | X_1 \oplus \dots \oplus X_{i} = 1] \cdot (1-a_i)\\
        &= (1-p_{i+1}) \cdot a_i + p_{i+1} \cdot (1-a_i)\\
        &= p_{i+1} + (1-2p_{i+1}) a_i\\
        &= \frac{1}{2} \left( 1 + \prod_{j=1}^{i+1} (1-2p_j) \right)
    \end{align*}
    This concludes the inductive case.
    Notice that $\frac{m-t}{n} \leq p_i \leq \frac{m}{n-t}$ since every time we choose an element, the proportional of special elements to total elements left is at least $(m-t)/n$ and at most $m/(n-t)$. Therefore,
    \[ \min_{\frac{m-t}{n} \leq p_i \leq \frac{m}{n-t}} \frac{1}{2} \left( 1 + \prod_{i=1}^t (1-2p_i) \right) \leq {\operatorname{Pr}}[\text{$X$ is even}] \leq \max_{\frac{m-t}{n} \leq p_i \leq \frac{m}{n-t}} \frac{1}{2} \left( 1 + \prod_{i=1}^t (1-2p_i) \right) \ .\] The result follows by algebraic manipulation.
\end{proof}

\xorBiasHypCor*
\begin{proof}
     Observe that
     \[ \max_{\frac{m-t}{n} \leq p_i \leq \frac{m}{n-t}} \prod_{i=1}^t (1-2p_i) \leq \max_{\frac{m-t}{n} \leq p_i \leq \frac{m}{n-t}} \prod_{i=1}^t |1-2p_i| \leq |1-2p_i|^t \ . \]
     By \cref{lem: xor_bias_no_replacement},
     \begin{equation*}
         {\operatorname{Pr}}[\text{$X$ is even}]
         = \frac{1}{2} + \frac{1}{2} \max_{\frac{m-t}{n} \leq p_i \leq \frac{m}{n-t}} \prod_{i=1}^t (1-2p_i)
         \leq \frac{1}{2} + \frac{1}{2} |1-2p|^t \ .
     \end{equation*}
     Similarly,
     \begin{align*}
         {\operatorname{Pr}}[\text{$X$ is even}]
         &\geq \frac{1}{2} + \frac{1}{2} \min_{\frac{m-t}{n} \leq p_i \leq \frac{m}{n-t}} \prod_{i=1}^t (1-2p_i)\\
         &= \frac{1}{2} + \frac{1}{2} \max_{\frac{m-t}{n} \leq p_i \leq \frac{m}{n-t}} -\prod_{i=1}^t (1-2p_i)\\
         &= \frac{1}{2} - \frac{1}{2} \max_{\frac{m-t}{n} \leq p_i \leq \frac{m}{n-t}} \prod_{i=1}^t (1-2p_i)\\
         &\geq \frac{1}{2} - \frac{1}{2} |1-2p|^t \ .
     \end{align*}
     This implies that the probability that $X$ is odd is at most $1/2 + (1/2) |1-2p|^t$.
\end{proof}
}


\end{document}